\documentclass[amscd,amsmath,amssymb,verbatim, aps, prd, showpacs]{revtex4-2}
\usepackage{amsmath, amsthm, amssymb}

\usepackage{graphicx}
\usepackage{hyperref}
\usepackage{textcomp}
\usepackage[OT1,T1]{fontenc}
\usepackage{epstopdf}
\newtheorem{theorem}{Theorem}
\newtheorem{lemma}[theorem]{Lemma}
\newtheorem{corollary}[theorem]{Corollary}
\newtheorem{proposition}[theorem]{Proposition}

\newtheorem{definition}{Definition}
\newtheorem{remark}{Remark}
\newcommand{\wolframcharacter}[1]{}
\newcommand{\beginwidetext}{\begin{widetext}}
\newcommand{\mnras}{Monthly Notices of the Royal Astronomical Society}
\newcommand{\apjs}{Astrophysical Journal Supplement}
\newcommand{\aap}{Astronomy and Astrophysics}

\begin{document}

\title{Observables for the Effect of Gravity on Electromagnetic
Polarization}
\author{Kjell Tangen}
\email{kjell.tangen@gmail.com}
\affiliation{Oslo, Norway}\label{I1}
\date{\today}
\begin{abstract}
Does gravity affect the polarization of electromagnetic radiation
in an observable way? The effect of gravity on the observed polarization
of a ray of electromagnetic radiation is investigated for an arbitrary
4-dimensional spacetime and radiation with a frequency spectrum
within the geometric optics limit and with arbitrary state of polarization.
Focusing on effects observable by a single inertial observer, we
show how the presence of curvature along the null geodesic of polarized
electromagnetic radiation may induce observable changes in the state
of polarization. We find a set of scalars that quantify the effect
and derive their transport equations. Two of these scalars, the
{\itshape polarization degree} and the {\itshape circular polarization
degree}, are polarization state observables that are conserved along
the radiation geodesic. Four observables that quantify time rate
of change of the observed state of polarization are identified.
One observable, the {\itshape polarization axis rotation rate},
measures angular motion of the polarization axis. The presence of
curvature along the null geodesic of polarized electromagnetic radiation
may induce rotation of the polarization axis of the radiation, as
measured by an inertial observer. These observables quantify gravitational
effects on electromagnetic polarization unambiguously. They are
generally sensitive to gravity, which implies that polarized electromagnetic
radiation can potentially be used to measure the gravitational fields
it traverses. These observables and their corresponding transport
equations provide a complete representation of how gravity affects
the observed state of polarization of electromagnetic radiation
with frequencies above the geometric optics limit. Polarization
wiggling is sourced by curvature twist, which is a scalar derived
from the Riemann tensor. Curvature twist is closely related to the
magnetic part of the Weyl tensor, the second Weyl scalar as well
as the rotation of the rest frame geodesic congruence. The results
of this paper are valid for any metric theory of gravity.
\end{abstract}
\maketitle

 {\bfseries Keywords} \ Electromagnetic polarization, Observable,
Gravity, Gravitational Faraday Effect

\section{Introduction}

Does gravity affect the polarization of electromagnetic radiation
in a measurable way? Gravity is a complex phenomenon with widely
different physical manifestations on a wide range of length scales,
ranging from scalar gravity, which drives structure formation at
the largest observable scales in the universe, to tensorial gravity
causing displacements at the smallest of scales probed by human
experiment \cite{Schultz-2018}. Just as gravity's effects on matter
vary between the well-known, well described by Newtonian gravity,
to the faint effects of gravitational waves, just recently probed
experimentally \cite{Abbot-2016}, gravity's effects on electromagnetic
radiation vary between the well understood and the elusive. Among
the well understood effects are gravitational redshift and lensing.
Among the elusive effects is gravity's direct effect on electromagnetic
polarization. 

The effect of a gravitomagnetic field on electromagnetic polarization
has been studied extensively, and is known as the {\itshape Gravitational
Faraday Effect }\cite{Dehnen-1973} because of its analogy with the
Faraday Effect, which is the effect an ambient magnetic field has
on the polarization of a passing ray of electromagnetic radiation.
Still, from a theoretical perspective, an observable quantifying
the effect is missing. The observability of the effect has been
questioned \cite{Faraoni:2007uh}. Within cosmology, there is an
established presumption that gravity has no observable direct effect
on the polarized Cosmic Microwave Background (CMB) \cite{Dodelson-2003}.

The Gravitational Faraday Effect emerges when quantifying the evolution
of the polarization vector along the trajectory of a plane electromagnetic
wave propagating in a gravitational field. It appears as a rotation
of the polarization vector as it propagates along the radiation
trajectory. Quantifying this rotation requires a choice of reference
frame at each point along the radiation geodesic. Since there is
no canonical choice of reference frames in a curved spacetime, quantifying
the effect of a gravitational field on the polarization of electromagnetic
radiation by quantifying the rotation of the Gravitational Faraday
Effect yields ambiguous results.

From an observational perspective, the Gravitational Faraday Effect
on a single ray of radiation provides an insufficient view of the
effect of gravity on electromagnetic polarization since it lacks
unambiguous quantification. A different view and representation
of this effect may be needed to uncover observables for gravitational
effects on electromagnetic polarization. Such observables may instead
be obtained by studying how optical parameters vary over a cross
section of a bundle of rays of radiation. There are interesting
examples of observables quantifying gravitational effects on electromagnetic
radiation that were uncovered this way. One such example is how
gravity affects the {\itshape Orbital Angular Momentum} (OAM) of
a beam of electromagnetic radiation. OAM is the angular momentum
of a beam of electromagnetic radiation around the center of the
beam \cite{Allen-1992}. OAM is observable and can be quantified
from the phase variation of the radiation over a spatial cross-section
of the beam. It is sensitive to gravity \cite{Tamburini-2011}, and
OAM has recently been used to quantify the rotation parameters of
the M87 black hole \cite{Tamburini-2020,Tamburini-2021}. A second
example is how the variation over the sky of linear polarization
of the CMB has been used to probe the early universe \cite{WMAP-2013,Planck-2018}.
Also, gravitational vector modes may induce angular correlations
in the polarized CMB as well as in radiation from quasars \cite{Morales-2007}.
These are examples of how a gravitational field may induce variations
of optical parameters over a spatial cross section of a radiation
field and how these variations can be used to probe the gravitational
fields inducing them. 

In this paper, we will do something similar. But instead of considering
parameter variations over a spatial cross section of a congruence
of null geodesics, we will consider the intersection of a null geodesic
congruence of electromagnetic radiation by the timelike geodesic
of an inertial observer and quantify how observable optical parameters
vary with proper time of this observer. The present paper provides
a view and representation of the effect of gravity on electromagnetic
polarization that is different from that of the Gravitational Faraday
Effect. It represents the effect from the perspective of a single
inertial observer: Given a pointlike inertial emitter of polarized
electromagnetic radation, how is the presence of a gravitational
field along the trajectory of the radiation manifested in observations
made by this observer? Given timelike geodesics of the emitter and
observer, we will identify observables that quantify this effect
unambiguously. The effect is referred to as {\itshape polarization
wiggling}, which alludes to how the observed polarization axis of
any polarized electromagnetic radiation will wiggle in the presence
of gravitational fields along the radiation geodesics. 

The objective of this paper is simply to answer the question posed
initially for an arbitrary 4-dimensional spacetime. We will show
that gravity indeed has a real, direct physical effect on electromagnetic
polarization by identifying observables for this effect. By using
covariant techniques that are valid for any metric theory of gravity,
we make sure that the results are valid for any 4-dimensional spacetime
geometry defined by a pseudo-Riemannian metric. This effect opens
the prospect of measuring gravitational fields directly using polarized
electromagnetic radiation. 

A 4-dimensional metric theory of gravity is assumed. Our metric
convention is $(-+++)$, and relativistic units with $c=1$ are used.
The Einstein summation convention is assumed on repeated indices.
Throughout the paper, we follow a convention where, unless otherwise
stated, geodesics and geodesic congruences are referenced by their
tangent vector fields.

\subsection{The Gravitational Faraday Effect}\label{XRef-Subsection-66111547}

Based on theoretical work on geometric optics in gravitational fields
by Rytov \cite{rytov1937wave}, Skrotskii first described the effect
of a stationary gravitational field on the polarization of a plane
electromagnetic wave \cite{Skrotskii-1957}. This effect is now alternately
referred to as the {\itshape Skrotskii effect} or the {\itshape
Gravitational Faraday effect}. It should be mentioned that the term
Gravitational Farday Effect alludes to the analogy with the Faraday
Effect and the corresponding deep analogy between electromagnetism
and the gravitoelectromagnetic formulation of General Relativity
\cite{Poisson-Will-2014,Mashhoon:2003ax,Carrol-2004}. Since Skrotskii's
original paper, the effects of gravity on polarized electromagnetic
radiation has been studied by many authors \cite{Dehnen-1973,Mashhoon-1973,Perlick-1993,Balazs-1958,Plebanski-1960,Godfrey-1970,Pineault-1977,Su-Mallet-1980,Fayos-Llosa-1982,Ishihara-1988,Gnedin-Dymnikova-1988,Kobzarev-1988,Kopeikin-2001,Nouri-Zonoz-1999,Sereno-2004,Sereno-2005,Brodutch-2011-11,Brodutch-2011-12,Montanari-1998,Ruggiero-2007,Morales-2007}.
All of these papers study gravitational effects on plane electromagnetic
waves. Most papers focus on the effect of a stationary gravitational
field \cite{Mashhoon-1973,Balazs-1958,Plebanski-1960,Godfrey-1970,Pineault-1977,Su-Mallet-1980,Fayos-Llosa-1982,Ishihara-1988}\cite{Nouri-Zonoz-1999,Brodutch-2011-11,Brodutch-2011-12}.
In particular, Fayos and Llosa analyzed the Gravitational Faraday
effect for an arbitrary stationary spacetime \cite{Fayos-Llosa-1982}.
Their result is quite general and reproduces earlier results obtained
by Plebanski \cite{Plebanski-1960}, Godfrey \cite{Godfrey-1970}
and Pineault and Roeder \cite{Pineault-1977} for more specific stationary
metrics. More general gravitational systems in asymptotically flat
spacetimes have also been studied \cite{Plebanski-1960}\cite{Gnedin-Dymnikova-1988,Kobzarev-1988,Kopeikin-2001}\cite{Sereno-2004,Sereno-2005}.
Montanari studied the effect of a plane gravitational wave \cite{Montanari-1998}.
Perlick and Hasse studied the effect in conformally stationary spacetimes
\cite{Perlick-1993}. Ruggiero and Tartaglia analyzed the Gravitational
Faraday effect of the gravitational field of a binary pulsar system
\cite{Ruggiero-2007}. Morales and S\`aez studied the effect of gravitational
vector perturbations in a cosmology with a flat FLRW background
\cite{Morales-2007}.

Let us give a brief summary of the Gravitational Faraday effect.
The Gravitational Faraday effect appears in linear gravity when
expanding the parallel transport equation for the polarization vector
of the electromagnetic vector potential. Given a gravitational field
expressed as a metric perturbation $h_{\mu \nu }$ to a Minkowsski
background, this parallel transport equation can be rewritten as
a transport equation for the polarization vector of the electric
field. To first perturbative order, the transport equation for the
transverse polarization vector $\epsilon ^{i}$ is, stated in an
inertial frame and parametrized in terms of distance $s$ along the
radiation path \cite{Kopeikin-2001}: 
\begin{equation}
\frac{d\text{\boldmath $\epsilon $}}{ds}=\text{\boldmath $\Omega
$}\times \text{\boldmath $\epsilon $},%
\label{XRef-Equation-121771546}
\end{equation}

where $\Omega ^{i}\equiv \frac{1}{2p}\varepsilon _{\mathrm{ijk}}p^{\alpha
}h_{\alpha  j,k}$, $p$ is the photon energy, $\varepsilon _{\mathrm{ijk}}$
is the antisymmetric Levi-Civita symbol, and $p^{\alpha }$ is the
photon 4-momentum.

Eq. (\ref{XRef-Equation-121771546}) exposes the Gravitational Faraday
effect as a rotation of the polarization vector $\text{\boldmath
$\epsilon $}$ as the radiation propagates along its trajectory.
The angular velocity is $\text{\boldmath $\Omega $}$. Projecting
$\text{\boldmath $\Omega $}$ against the direction of propagation,
$\hat{\text{\boldmath $p$}}$, gives the rotation rate of the polarization
axis about the direction of propagation: 
\begin{equation}
\nu \equiv {\hat{p}}^{i}\Omega ^{i}=\frac{1}{2p}{\hat{p}}^{i}\varepsilon
_{\mathrm{ijk}}p^{\alpha }h_{\alpha  j,k}.%
\label{XRef-Equation-177445}
\end{equation}

If we consider the example of a stationary metric with $w_{i}\equiv
h_{0i}\neq 0$ and $h_{\mathrm{ij}}=2\Phi  \delta _{\mathrm{ij}}$,
representing the gravitational field of a rotating body, the rotation
rate $\nu $ is
\begin{equation}
\nu =-\frac{1}{2}\hat{\text{\boldmath $p$}}\cdot \text{\boldmath
$B$},%
\label{XRef-Equation-11972822}
\end{equation}

where $\text{\boldmath $B$}\equiv \nabla \times \text{\boldmath
$w$}$ is the gravitomagnetic field \cite{Poisson-Will-2014}. This
relates the rotation rate of the polarization axis to the strenght
of the gravitomagnetic field, which motivates naming the effect
the Gravitational Faraday Effect.

The Gravitational Faraday effect can easily be quantified, as was
done here, by choosing a polarization basis at each point along
the geodesic, then evaluating the rotation rate of the polarization
vector along the geodesic at each point and finally integrating
along the null geodesic to obtain the total rotation of the polarization
vector over the path. However, as was pointed out by Brodutch, Demarie
and Terno \cite{Brodutch-2011-11}, the resulting rotation is only
meaningful in the context of a particular choice of polarization
basis at each point along the geodesic. Since there is no canonical
choice of reference frames in a curved spacetime, different choices
of polarization basis along the null geodesic may yield different
results for the angle of rotation of the polarization vector. 

In the literature on the Gravitational Faraday effect, authors have
made different choices of polarization basis along the null geodesic.
While Skrotskii calculated the rotation rate relative to a Frenet
frame, some authors chose to analyze the Gravitational Faraday Effect
in the coordinate frame \cite{Plebanski-1960,Kopeikin-2001}. Other
authors have computed the rotation rate relative to a frame that
is parallel propagated along the spatial trajectory of the ray by
using the spatial metric \cite{Fayos-Llosa-1982,Nouri-Zonoz-1999,Brodutch-2011-11}.
Each of these approaches quantifies the Gravitational Faraday Effect
differently, and the results are therefore difficult to compare.

An additional complicating factor is that observable quantities
also will depend on the motions of the frames of emission and observation.
If we assume inertial frames of emission and observation, these
frames will for instance be subject to frame dragging in the presence
of gravitomagnetic fields, which will affect the observed rotation
rates of the polarization axis of the radiation. It is therefore
essential that the motions of the frames of observation and emission
are taken into account in the formulation of observables that quantify
the Gravitational Faraday effect.

\subsection{Paper Overview}

The main objective of this paper is to show that gravity indeed
has a real, physical effect on electromagnetic polarization by identifying
observables that quantify this effect unambiguously, independent
of any choice of local polarization basis, for any 4-dimensional
spacetime and for radiation with arbitrary frequency composition
and state of polarization.

To meet this objective, we have chosen a perspective on the effects
of gravity on electromagnetic polarization that differs from that
of previous publications on this topic. While previous authors have
restricted their attention to specific gravitational systems using
specific metrics, we will approach the effect geometrically and
in complete generality, using covariant constructs.

To present the results in a proper way, the paper gives a comprehensive
and self-contained account of gravity's effect on the polarization
of electromagnetic radiation. Each of the main constructs are defined
in a labeled Definition statement. The results are derived step
by step through a long series of minor propositions. The main results
are formulated as theorems. Although somewhat elaborate, this way
of presenting the results was chosen to make the derivations explicit
and easier to follow. 

First, the nomenclature and covariant constructs needed for our
study is introduced in Section \ref{XRef-Section-44183620}.

In sections \ref{XRef-Section-331204723} through \ref{XRef-Section-112692639},
the results are derived step by step. For readers just interested
in the results, a brief result summary can be found in Section \ref{XRef-Section-527172254}.\ \ 

In Section \ref{XRef-Section-331204723}, we define the problem formally
and covariantly for a plane electromagnetic wave, using standard
Riemannian geometry. We seek a unified treatment of the effect,
valid for any 4-dimensional spacetime. Two complex scalars, which
we refer to as the {\itshape polarization wiggle scalars}, are defined
as two independent transverse projections of the time derivative
of the transverse polarization vector. Their covariant transport
equations are derived. By performing a helicity decomposition of
the polarization wiggle scalars, a set of four observables that
quantify time rates of change of the state of polarization emerge.
One of these observables,\ \ the {\itshape polarization axis wiggle
rate}, is the angular speed of the polarization axis, as measured
by an inertial observer. Another observable, the {\itshape circular
polarization wiggle rate,} is proportional to the rate of change
of the circular polarization degree. At the end of this section,
we demonstrate how the polarization axis wiggle rate relates to
previous results on the Gravitational Faraday Effect in linear gravity.

The results are generalized to composite radiation with a frequency
spectrum within the geometric optics limit and arbitrary state of
polarization in Section \ref{XRef-Section-52794743}. We find that
the {\itshape total degree of polarization}, one of the observables
of the state of polarization, is conserved along the null geodesic
of the radiation.

In Section \ref{XRef-Section-112692639}, we study the curvature
twist scalar, which appears in the transport equation for the polarization
wiggle scalars. We relate it to several known quantities, such as
the second Weyl scalar, the magnetic part of the Weyl tensor as
well as the rotation of the rest frame geodesic congruence.

A summary of the results of this paper is given in Section \ref{XRef-Section-527172254}.
A discussion of the results can be found in Section \ref{XRef-Section-65111418},
and conclusions are drawn in Section \ref{XRef-Section-457621}.

\section{Definitions}\label{XRef-Section-44183620}

In this section, we will review the nomenclature and basic constructs
that will be applied in later sections. Since some of the constructs
used are new or not commonly used, they will be reviewed here before
we commence our analysis in the subsequent section. More details
can be found in the appendices.

\subsection{Spacetime Decomposition}\label{XRef-Subsection-2269858}

Throughout the paper, we will assume the presence of a congruence
of\ \ timelike geodesics $u^{\mu }$. These geodesics will be referred
to as the timelike{\itshape  reference geodesics} and the congruence
as the {\itshape reference geodesic congruence}. They can be viewed
as the geodesics of inertial objects {\itshape  at rest}, and an
inertial frame along one such geodesic is called a {\itshape rest
frame}. In Gaussian normal coordinates (synchronous coordinates)
adapted to the timelike reference geodesics, objects at rest are
stationary.

\,Given a spacetime with metric $g_{\mu \nu }$, the spatial projection
tensor is defined as
\begin{equation}
\gamma _{\mu \nu }\equiv g_{\mu \nu }+u_{\mu }u_{\nu }.%
\label{XRef-Equation-2268398}
\end{equation}

Then, assuming the presence of an electromagnetic plane wave with
wave vector $p^{\mu }$, at any point along its null geodesic, the
null vector $p^{\mu }$ can be decomposed into timelike and spacelike
vectors as follows: 
\begin{equation}
p^{\mu }=p( u^{\mu }+{\hat{p}}^{\mu }) ,%
\label{XRef-Equation-62782635}
\end{equation}

where $p\equiv -u_{\mu }p^{\mu }$ is the photon energy measured
in the rest frame, while ${\hat{p}}^{\mu }$ is the vector
\begin{equation}
{\hat{p}}^{\mu }\equiv \frac{p^{\mu }}{p}-u^{\mu }=\gamma _{\nu
}^{\mu }\frac{p^{\nu }}{p}.%
\label{XRef-Equation-910163241}
\end{equation}

${\hat{p}}^{\mu }$ is a spacelike unit vector. It is orthogonal
to $u^{\mu }$, and represents the spatial direction of propagation
of the radiation. We will refer to $p$ as the {\itshape rest frame
momentum}. The rest frame momentum can also be found by contracting
${\hat{p}}^{\mu }$ with the 4-momentum: $p=p_{\mu }{\hat{p}}^{\mu
}$. \,We will refer to the null vector $n^{\mu }\equiv p^{\mu }/p=u^{\mu
}+{\hat{p}}^{\mu }$ as the {\itshape scale invariant null vector}
associated with a null vector $p^{\mu }$.

The state of polarization of electromagnetic radiation can be quantified
in terms of the four Stokes parameters. They are measurable quantities
that can be expressed in terms of the electric field vector and
its projections against a {\itshape polarization basis}, which consists
of two mutually orthogonal unit vectors that both are transverse
to the direction of propagation of the radiation, $\hat{p}$. Let
the polarization basis vectors be denoted by $e_{A}^{\mu }, A=1,2$.
In the following, we will assume that the polarization basis is
parallel transported along the timelike reference geodesics. The
four unit vectors $u^{\mu }, {\hat{p}}^{\mu }$ and $e_{A}^{\mu },
A=1,2$ form a tetrad.

This decomposition can be made for any spacetime geometry, and does
not put any constraints on the spacetime geometry.

The {\itshape null tetrad formalism} introduced by Newman and Penrose
\cite{Newman-Penrose-1962} represents an alernative spacetime decomposition.
Here we will just capture the basic representations needed for this
paper. It is referred to the book by Stephani et. al for a complete
account of the formalism \cite{Stephani-2003}. A complex null tetrad
consists of four null vectors; two real and two complex. Our choice
of null tetrad can be expressed in terms of the $(u, \hat{p}, e_{A})$
tetrad as follows:
\begin{gather*}
k^{\mu }\equiv \frac{1}{\sqrt{2}}\left( u^{\mu }+{\hat{p}}^{\mu
}\right) , l^{\mu }\equiv \frac{1}{\sqrt{2}}\left( u^{\mu }-{\hat{p}}^{\mu
}\right) 
\\m_{+}^{\mu }\equiv \frac{1}{\sqrt{2}}\left( e_{1}^{\mu }+i e_{2}^{\mu
}\right) ,\ \ m_{-}^{\mu }\equiv \frac{1}{\sqrt{2}}\left( e_{1}^{\mu
}-i e_{2}^{\mu }\right) .
\end{gather*}

This basis satisfies $k^{\mu }k_{\mu }=l^{\mu }l_{\mu }=m_{+}^{\mu
}m_{\mu }^{+}=m_{-}^{\mu }m_{\mu }^{-}=0$, $k^{\mu }l_{\mu }=-1$
and $m_{+}^{\mu }m_{\mu }^{-}=1$.

Define differential operators $D\equiv k^{\alpha }\nabla _{\alpha
}$ and $\Delta \equiv l^{\alpha }\nabla _{\alpha }$. We notice that
\[
\frac{1}{\sqrt{2}}\left( D+\Delta \right) =\nabla _{u}.
\]

We will assume that the polarization basis $e_{A}^{\mu }$ is parallel
transported along the geodesic $u^{\mu }$ and therefore satisfies
$\nabla _{u}e_{A}^{\mu }=0$. Therefore, 
\begin{equation}
\left( D+\Delta \right) m_{\pm }^{\nu }=\sqrt{2}\nabla _{u}m_{\pm
}^{\nu }=0.%
\label{XRef-Equation-102565226}
\end{equation}

\subsection{Screen Projector}

The spatial projection tensor $\gamma _{\mu \nu }$ can be further
decomposed in the presence of a null congruence. 
\begin{definition}

Given a null vector $p^{\mu }$ and a timelike unit vector $u^{\mu
}$, define the {\bfseries screen projector} as the transverse projector\label{XRef-Definition-122074554}
\begin{equation}
S_{\mu \nu }\equiv \gamma _{\mu \nu }-{\hat{p}}_{\mu }{\hat{p}}_{\nu
}=g_{\mu \nu }+u_{\mu }u_{\nu }-{\hat{p}}_{\mu }{\hat{p}}_{\nu }.
\end{equation}
\end{definition}

$S_{\mu \nu }$ is referred to as the {\itshape screen projector
}\cite{Naruko-2013,Pitrou-2021}, because it can be used to project
any 4-field to a transverse 4-field. If $v^{\mu }$ is an arbitrary
4-vector, $S_{\nu }^{\mu }v^{\nu }$ will be transverse to $p^{\mu
}$, because $u_{\mu }S_{\nu }^{\mu }={\hat{p}}_{\mu }S_{\nu }^{\mu
}=0$. Furthermore, the screen projector satisfies
\[
S_{\alpha }^{\mu }S_{\nu }^{\alpha }=S_{\nu }^{\mu }.
\]

The screen projector has trace 2: $g^{\mu \nu }S_{\mu \nu }=2$.
This decomposition can be done for any spacetime geometry.

Notice that, if $u^{\mu }$ belongs to a timelike geodesic congruence
that is hypersurface orthogonal, $\gamma _{\mu \nu }$ represents
the metric of a 3-dimensional spatial hypersurface, while if $p^{\mu
}$ belongs to a hypersurface orthogonal null congruence, $S_{\mu
\nu }$ represents the metric of a 2-dimensional hypersurface everywhere
transverse to the direction of propagation ${\hat{p}}^{\mu }$.

In terms of the\ \ $(k,l,m_{\pm })$ null tetrad, the screen projector
takes the form
\[
S_{\nu }^{\mu }=\delta _{\nu }^{\mu }+k^{\mu }l_{\nu }+l^{\mu }k_{\nu
}.
\]

\subsection{Transverse Directional Derivative}
\begin{definition}

Given a null vector $p^{\mu }$ and its screen projector $S_{\mu
\nu }$, define the {\bfseries {\itshape transverse covariant derivative}}
of a vector $V^{\mu }$ as 
\[
\mathcal{D}_{\mu }V^{\nu }\equiv S_{\alpha }^{\nu }\nabla _{\mu
}V^{\alpha }.
\]
\end{definition}
\begin{definition}

Corresponding to the directional derivative $\nabla _{U}V$ of the
vector $V$ with respect to a vector $U$, define the {\bfseries {\itshape
transverse covariant directional derivative}} of $V$ with respect
to $U$ as
\[
\mathcal{D}_{U}V^{\mu }\equiv u^{\beta }\mathcal{D}_{\beta }V^{\mu
}=S_{\alpha }^{\mu }\nabla _{U}V^{\alpha }=S_{\alpha }^{\mu }U^{\beta
}\nabla _{\beta }V^{\alpha }.
\]
\end{definition}

These definitions trivially extend to tensors of higher rank and
scalars. For example, for a $(0,n)$ tensor $T$, the transverse covariant
derivative of $T$ is defined by
\[
\mathcal{D}_{\lambda }T_{\mu _{1}...\mu _{n}}\equiv S_{\mu _{1}}^{\alpha
_{1}}...S_{\mu _{n}}^{\alpha _{n}}\nabla _{\lambda }T_{\alpha _{1}...\alpha
_{n}},
\]

while the transverse covariant directional derivative of $T$ with
respect to $U$ is defined by
\[
\mathcal{D}_{U}T_{\mu _{1}...\mu _{n}}\equiv U^{\lambda }\mathcal{D}_{\lambda
}T_{\mu _{1}...\mu _{n}}= S_{\mu _{1}}^{\alpha _{1}}...S_{\mu _{n}}^{\alpha
_{n}}U^{\lambda }\nabla _{\lambda }T_{\alpha _{1}...\alpha _{n}}.
\]

The transverse covariant derivative of a scalar is simply defined
as its derivative: $\mathcal{D}_{\mu }S\equiv \nabla _{\mu }S$ for
scalar $S$.
\begin{proposition}

The transverse covariant derivative of the screen projector vanishes:\label{XRef-Proposition-11236541}
\begin{equation}
\mathcal{D}_{\mu }S_{\alpha \beta }\equiv S_{\alpha }^{\lambda }S_{\beta
}^{\gamma }\nabla _{\mu }S_{\lambda \gamma }=0.
\end{equation}
\end{proposition}
\begin{proof}

Since the covariant derivative of the metric vanishes and $u^{\mu
}S_{\mu \nu }={\hat{p}}^{\mu }S_{\mu \nu }=0$, it follows directly
from its definition in Definition \ref{XRef-Definition-122074554}
that the transverse covariant derivative of the screen projector
vanishes.
\end{proof}

\subsection{Screen Rotator}\label{XRef-Subsection-43010640}
\begin{definition}

Given a null vector $p^{\mu }$ and a timelike unit vector $u^{\mu
}$, define the {\bfseries {\itshape screen rotator}} as the antisymmetric
tensor\label{XRef-Definition-112264556}
\begin{equation}
\epsilon _{\mu \nu }\equiv u^{\alpha }{\hat{p}}^{\beta }\epsilon
_{\alpha \beta \mu \nu },%
\label{XRef-Equation-11237352}
\end{equation}

where $\epsilon _{\alpha \beta \mu \nu }$ is the Levi-Civita tensor.
\end{definition}

$\epsilon _{\mu \nu }$ is transverse by construction, because $\epsilon
_{\mu \nu }u^{\nu }=\epsilon _{\mu \nu }{\hat{p}}^{\nu }=\epsilon
_{\mu \nu }p^{\mu }=0$. We refer to this tensor as the screen rotator,
because it is a generator of rotations in the transverse plane.
A small rotation $R_{\mu \nu }( \varphi ) $ an angle $\varphi $
about the propagation direction ${\hat{p}}^{\mu }$ can, when expanded
to first order in $\varphi $, be written as
\begin{equation}
R_{\mu \nu }( \varphi ) =S_{\mu \nu }-\varphi  \epsilon _{\mu \nu
}+\mathcal{O}( \varphi ^{2}) .%
\label{XRef-Equation-225155538}
\end{equation}

In terms of the\ \ $(k,l,m_{\pm })$ null tetrad, the screen rotator
takes the form
\[
\epsilon ^{\mu \nu }=i( m_{+}^{\mu }m_{-}^{\nu }-m_{-}^{\mu }m_{+}^{\nu
}) .
\]
\begin{proposition}

The screen rotator satisfies\label{XRef-Proposition-122083225}
\begin{equation}
\epsilon _{\mu \nu }\epsilon ^{\rho \gamma }=S_{\mu }^{\rho }S_{\nu
}^{\gamma }-S_{\mu }^{\gamma }S_{\nu }^{\rho }.%
\label{XRef-Equation-1220121228}
\end{equation}
\end{proposition}
\begin{proof}

This follows from expanding the left-hand side of eq.(\ref{XRef-Equation-1220121228})
in a screen basis $e_{\mu }^{A}$, using that the 2-dimensional Levi-Civita
symbol satisfies $\varepsilon _{\mathrm{AB}}\varepsilon ^{\mathrm{CD}}=\delta
_{A}^{C}\delta _{B}^{D}-\delta _{A}^{D}\delta _{B}^{C}$ and finally
that the screen basis satisfies $e_{\mu }^{A}e_{A}^{\rho }=S_{\mu
}^{\rho }$.
\end{proof}
\begin{proposition}

The transverse covariant derivative of the screen rotator vanishes:\label{XRef-Proposition-112219738}
\begin{equation}
\mathcal{D}_{\alpha }\epsilon _{\mu \nu }=0.
\end{equation}
\end{proposition}
\begin{proof}

The proof is given in Appendix \ref{XRef-Subsection-22575619}.
\end{proof}

\subsection{The Polarization of a Plane Electromagnetic Wave}\label{XRef-Subsection-527133422}

Let us review the basic constructs of electromagnetism that we will
build upon. We will assume that the emitted electromagnetic radiation
satisfies the geometric optics approximation, which means that the
wavelength of this\ \ radiation is much smaller than any relevant
length scale of the gravitational field. In this approximation,
referred to as the {\itshape geometric optics limit}, a plane electromagnetic
wave can be represented by a vector potential \cite{Misner}
\begin{equation}
A_{\mu }=a_{\mu }e^{i \vartheta },%
\label{XRef-Equation-224184018}
\end{equation}

where $a_{\mu }$ is a slowly varying complex amplitude, while $\vartheta
$ is a rapidly varying real phase. By using a complex expression
for the vector potential in eq. (\ref{XRef-Equation-224184018}),
we apply the\ \ common convention that the physical quantity can
be found by taking the real part of the expression. This convention
also applies to any quantity derived linearly from\ \ the vector
potential. The wave vector, corresponding to the photon 4-momentum,
is defined as $p_{\mu }=\nabla _{\mu }\vartheta $. By applying the
vector potential of eq. (\ref{XRef-Equation-224184018}) to the Maxwell
equations, it can be shown that $p^{\mu }$ is a null vector. Furthermore,
$p^{\mu }$ satisfies the geodesic equation: $p^{\alpha }\nabla _{\alpha
}p^{\mu }=0$. The scalar amplitude is $a\equiv \sqrt{a_{\mu }^{*}a^{\mu
}}$. $a$ is referred to as the {\itshape magnitude} of the field.
The polarization vector of the vector potential is $f^{\mu }\equiv
a^{\mu }/a$. Hence, $f^{\mu }$ is a complex unit vector $(f_{\mu
}^{*}f^{\mu }=1)$, and it satisfies $p^{\mu }f_{\mu }=0$ \cite{Misner}.
The polarization vector $f^{\mu }$ is parallel transported along
the null geodesic and therefore satisfies $p^{\alpha }\nabla _{\alpha
}f^{\mu }=0$.

In the geometric optics limit, the field tensor for a single plane
wave takes the form
\begin{equation}
F_{\mu \nu }= i a( p_{\mu }f_{\nu }-p_{\nu }f_{\mu }) e^{i \vartheta
}.
\end{equation}
\begin{remark}

\end{remark}

Given a timelike geodesic vector $u^{\mu }$, the electric vector
$E^{\mu }$ is defined as
\begin{equation}
E^{\mu }\equiv u_{\alpha }F^{\alpha \mu }.%
\label{XRef-Equation-51511389}
\end{equation}
\begin{definition}

Given an electromagnetic plane wave in the geometric optics limit
with a vector potential represented by eq. (\ref{XRef-Equation-224184018}),
define the {\bfseries transverse polarization vector} $\epsilon
^{\mu }$ as the transverse projection of the vector potential polarization
vector, $f^{\mu }$:\label{XRef-Definition-121084119}
\[
\epsilon ^{\mu }\equiv  S_{\nu }^{\mu }f^{\nu }.
\]
\end{definition}

$\epsilon ^{\mu }$ is a complex, spacelike unit vector: $\epsilon
_{\mu }^{*}\epsilon ^{\mu }=1$. Hence, the electric field of an
electromagnetic plane wave can then be expressed in terms of the
transverse polarization vector as
\[
E^{\mu }=-i a p \epsilon ^{\mu }e^{i \vartheta }.
\]
\begin{proposition}

The transverse covariant directional derivative of the transverse
polarization vector $\epsilon ^{\mu }$ along the null geodesic $p^{\mu
}$ vanishes:\label{XRef-Proposition-112363540}
\begin{equation}
\mathcal{D}_{p}\epsilon ^{\mu }=0.%
\label{XRef-Equation-1122174843}
\end{equation}
\end{proposition}
\begin{proof}

From parallel propagation of the polarization vector $f^{\mu }$
follows that the transverse polarization vector $\epsilon ^{\mu
}$ evolves according to the equation
\[
p^{\alpha }\nabla _{\alpha }\epsilon ^{\mu }=\frac{p^{\mu }}{p}\epsilon
_{\beta }p^{\alpha }\nabla _{\alpha }u^{\beta }.
\]By applying the screen projector to this equation, we find that
the transverse directional derivative of $\epsilon ^{\mu }$ along
the null geodesic $p^{\mu }$ vanishes.
\end{proof}

\subsection{Covariant Representation of Electromagnetic Polarization}\label{XRef-Subsection-52815235}

The polarization state of a classical field of electromagnetic radiation
is completely characterized by the four Stokes parameters \cite{Born-Wolf-Optics,Jackson-Electrodynamics}.
They are measurable quantities that can be expressed in terms of
the electric field vector and its projections against a polarization
basis consisting of two transverse, mutually orthogonal unit vectors.
We will apply a covariant representation of the state of polarization
of a ray of electromagnetic radiation. This covariant representation
of polarization is reviewed in Appendix \ref{XRef-AppendixSection-430133648},
so more detail can be found there. 
\begin{definition}

Given an electromagnetic radiation field with electric vector $E^{\mu
}$, the {\bfseries coherency tensor} is defined as \cite{Born-Wolf-Optics}
\[
I_{\mu \nu }\equiv \left\langle  E_{\mu }E_{\nu }^{*}\right\rangle
,
\]

where the expectation value $\langle E_{\mu }E_{\nu }^{*}\rangle
$ is assumed to be a time average over a large number of radiation
cycles.
\end{definition}

The {\itshape radiation intensity} is the trace of the coherency
tensor: $I\equiv g^{\mu \nu }I_{\mu \nu }=S^{\mu \nu }I_{\mu \nu
}$. The $Q$, $U$ and $V$ Stokes parameters can be extracted from
the coherency tensor by making different projections of it against
the polarization basis vectors. See Appendix \ref{XRef-AppendixSection-430133648}
for details.
\begin{definition}

Define the {\bfseries {\itshape relative coherency tensor}} as $i_{\mu
\nu }\equiv I_{\mu \nu }/I$, where $I\equiv g^{\mu \nu }I_{\mu \nu
}=S^{\mu \nu }I_{\mu \nu }$ is the radiation intensity.\label{XRef-Definition-121091947}
\end{definition}
\begin{proposition}

For a plane wave with transverse polarization vector $\epsilon ^{\mu
}$, the relative coherency tensor is.\label{XRef-Proposition-12764641}
\[
i_{\mu \nu }^{\mathrm{pw}}( \epsilon ) =\epsilon _{\mu }\epsilon
_{\nu }^{*}.
\]
\end{proposition}
\begin{definition}

 The {\bfseries {\itshape polarization tensor}} is defined as the
traceless part of the relative coherency tensor: \label{XRef-Definition-12771128}
\begin{equation}
P_{\mu \nu }\equiv i_{\mu \nu }-\frac{1}{2}S_{\mu \nu }.%
\label{XRef-Equation-112318139}
\end{equation}
\end{definition}
\begin{definition}

The {\bfseries {\itshape polarization degree }}{\bfseries  }$\mathcal{P}$
is defined from the polarization tensor as\label{XRef-Definition-12773823}
\[
\mathcal{P}^{2}=2P^{\nu \mu }P_{\mu \nu }=2P^{\mu \nu }P_{\mu \nu
}^{*}.
\]
\end{definition}

In Appendix \ref{XRef-AppendixSection-430133648} it is shown that
Definition \ref{XRef-Definition-12773823} is equivalent to the conventional
definition of polarization degree. 
\begin{definition}

Define the {\bfseries {\itshape circular polarization degree}} $\mathcal{V}$
from the antisymmetric part of the polarization tensor as\label{XRef-Definition-11674543}
\begin{equation}
\mathcal{V}=i \epsilon ^{\mu \nu }P_{\mu \nu }.%
\label{XRef-Equation-103064113}
\end{equation}
\end{definition}

 $\mathcal{V}$ is the relative Stokes parameter measuring the degree
of circular polarization of electromagnetic radiation. See Section
\ref{XRef-Subsection-4309528} of Appendix \ref{XRef-AppendixSection-430133648}
for details.
\begin{proposition}

The circular polarization degree of a plane electromagnetic wave
is conserved along its null geodesic:\label{XRef-Proposition-112225615}
\begin{equation}
\nabla _{p}\mathcal{V}^{\mathrm{pw}}=0.%
\label{XRef-Equation-11222535}
\end{equation}
\end{proposition}
\begin{proof}

Since $\nabla _{p}\mathcal{V}=\mathcal{D}_{p}\mathcal{V}$, the result
follows by differentiating the right-hand side of eq. (\ref{XRef-Equation-103064113})
with $\mathcal{D}_{p}$, then applying Proposition \ref{XRef-Proposition-112363540}
and Proposition \ref{XRef-Proposition-12764641}.
\end{proof}
\begin{definition}

We will refer to frequencies above the geometric optics limit as
{\bfseries {\itshape geometric optics frequencies}}, and a frequency
spectrum with negligible power below the geometric optics limit
will be referred to as a {\bfseries {\itshape geometric optics frequency
spectrum}}.
\end{definition}

Now, consider a ray of composite electromagnetic radiation with
a geometric optics spectrum and in an arbitrary state of polarization.
We will model this radiation as an ensemble of plane wave components
traversing the same null geodesic, hence sharing the same scale
invariant null vector $n^{\mu }=u^{\mu }+{\hat{p}}^{\mu }$, but
with different polarization vector $\epsilon ^{\mu }$. It is known
that for two different wave components with electric vectors $E_{(j)}^{\mu
}$ and $E_{(k)}^{\mu }$, when averaged over a large number of radiation
cycles, the expectation value $\langle E_{\mu }^{(j)}E_{\nu }^{(k)*}\rangle
=0$ for $j\neq k$ \cite{Born-Wolf-Optics}. Therefore, the coherency
tensor for the ensemble is
\begin{equation}
I_{\mu \nu }=\sum \limits_{j}I_{\mu \nu }^{\mathrm{pw}}( \epsilon
_{\left( j\right) }) =\sum \limits_{j}I_{\left( j\right) }^{\mathrm{pw}}i_{\mu
\nu }^{\mathrm{pw}}( \epsilon _{\left( j\right) }) ,%
\label{XRef-Equation-5232107}
\end{equation}

where $j$ is the ensemble index. $I_{(j)}^{\mathrm{pw}}\equiv {(a_{(j)}p_{(j)})}^{2}$
is the intensity of plane wave component $j$, with $a_{(j)}$ and
$p_{(j)}$ denoting the magnitude and rest frame momentum, respectively,
of plane wave component $j$ of the ensemble, while $i_{\mu \nu }^{\mathrm{pw}}(
\epsilon _{(j)}) =\epsilon _{\mu }^{(j)}\epsilon _{\nu }^{(j)*}$
represents the relative coherency tensor of plane wave component
$j$ with transverse polarization $\epsilon _{\mu }^{(j)}$. The total
intensity of the ray is the trace of the coherency tensor:
\begin{equation}
I=S^{\mu \nu }I_{\mu \nu }=\sum \limits_{j}I_{\left( j\right) }^{\mathrm{pw}}.%
\label{XRef-Equation-126172427}
\end{equation}

Hence, the total intensity of the ray is the sum of the intensities
of its components, as expected.
\begin{proposition}

Given a ray of electromagnetic radiation with a geometric optics
spectrum and arbitrary state of polarization, the coherency tensor
$I_{\mu \nu }$, the relative coherency tensor $i_{\mu \nu }$ and
the polarization tensor $P_{\mu \nu }$ are all transverse.\label{XRef-Proposition-121094841}
\end{proposition}
\begin{proof}

By representing the ray as an ensemble of plane waves, the coherency
tensor $I_{\mu \nu }$ is transverse by eq. (\ref{XRef-Equation-5232107})
and the fact that $i_{\mu \nu }^{\mathrm{pw}}( \epsilon _{(j)})
=\epsilon _{\mu }^{(i)}\epsilon _{\nu }^{(i)*}$ of each plane wave
component is transverse. It then follows from Definition \ref{XRef-Definition-121091947}
that the relative coherency tensor $i_{\mu \nu }$ is transverse.\ \ Finally,
the polarization tensor $P_{\mu \nu }$, as defined in Definition
\ref{XRef-Definition-12771128},\ \ is transverse, since $i_{\mu
\nu }$ and $S_{\mu \nu }$ are both transverse.
\end{proof}
\begin{proposition}

Given a ray of electromagnetic radiation with a geometric optics
frequency spectrum and arbitrary state of polarization, eigenvectors
of its coherency tensor $I_{\mu \nu }$ are transverse.\label{XRef-Proposition-121095312}
\end{proposition}
\begin{proof}

Given an eigenvector $X^{\mu }$ of $I_{\mu \nu }$, we have $X^{\mu
}\propto I^{\mu \nu }X_{\nu }$. Since $I^{\mu \nu }$ is transverse
by Proposition \ref{XRef-Proposition-121094841}, it follows that
$X^{\mu }$ also must be transverse.
\end{proof}

\section{The Effect of Gravity on the Polarization of a Plane Electromagnetic
Wave}\label{XRef-Section-331204723}

Let us turn our attention to the problem at hand. The aim of this
section is to quantify effects of gravity on the polarization of
a plane electromagnetic wave in general terms, valid for any 4-dimensional
spacetime. We will do this using basic Riemannian geometry.

\subsection{The Polarization Wiggle Equations}\label{XRef-Subsection-58122719}

One of the basic tenets of Riemannian geometry states that, given
three vector fields $X$, $Y$ and $Z$, these fields couple through
the Riemann tensor as follows \cite{Carrol-2004}:
\begin{equation}
R( X,Y) Z=\nabla _{X}\nabla _{Y}Z-\nabla _{Y}\nabla _{X}Z-\nabla
_{\left[ X,Y\right] }Z,%
\label{XRef-Equation-6692314}
\end{equation}

or written in terms of components:
\begin{equation}
{R^{\mu }}_{\nu \alpha \beta }X^{\alpha }Y^{\beta }Z^{\nu }=\nabla
_{X}\nabla _{Y}Z^{\mu }-\nabla _{Y}\nabla _{X}Z^{\mu }-\nabla _{\left[
X,Y\right] }Z^{\mu }.%
\label{XRef-Equation-22592559}
\end{equation}

Assuming the presence of a plane electromagnetic wave with wave
vector $p^{\mu }$, let us apply this equation to the transverse
electromagnetic polarization vector $\epsilon ^{\mu }$ with $X^{\alpha
}=p^{\alpha }$, $Y^{\beta }=u^{\beta }$ and $Z^{\nu }=\epsilon ^{\nu
}$. $u^{\beta }$ is assumed to be a timelike reference geodesic,
cf. Section \ref{XRef-Subsection-2269858}. The resulting equation
is
\[
\nabla _{p}\nabla _{u}\epsilon _{\mu }-\nabla _{u}\nabla _{p}\epsilon
_{\mu }-\nabla _{\left[ p,u\right] }\epsilon _{\mu }=R_{\mu \nu
\alpha \beta }p^{\alpha }u^{\beta }\epsilon ^{\nu }.
\]

Transverse projection of this equation yields
\begin{equation}
\mathcal{D}_{p}\mathcal{D}_{u}\epsilon _{\mu }-\mathcal{D}_{u}\mathcal{D}_{p}\epsilon
_{\mu }-\mathcal{D}_{\left[ p,u\right] }\epsilon _{\mu }=S_{\mu
}^{\lambda }R_{\lambda \nu \alpha \beta }p^{\alpha }u^{\beta }\epsilon
^{\nu }.%
\label{XRef-Equation-22593926}
\end{equation}

Let us evaluate $\mathcal{D}_{[p,u]}\epsilon _{\mu }$ first. From
eq. (\ref{XRef-Equation-22593549}) of Appendix \ref{XRef-Subsection-22593618},
the commutator $[p,u]$ is
\[
{\left[ p,u\right] }^{\beta }=-\frac{p^{\beta }}{p}\left( \nabla
_{u}p+\frac{\nabla _{p}p}{p}\right) +u^{\beta }\frac{\nabla _{p}p}{p}.
\]

Then, the transverse directional derivative of the transverse polarization
vector is
\[
\mathcal{D}_{\left[ p,u\right] }\epsilon ^{\mu }=-\frac{1}{p}\left(
\nabla _{u}p+\frac{\nabla _{p}p}{p}\right) \mathcal{D}_{p}\epsilon
^{\mu }+\frac{\nabla _{p}p}{p}\mathcal{D}_{u}\epsilon ^{\mu }.
\]

By Proposition \ref{XRef-Proposition-112363540}, $\mathcal{D}_{p}\epsilon
^{\mu }=0$, which implies that
\[
\mathcal{D}_{\left[ p,u\right] }\epsilon ^{\mu }=\frac{\nabla _{p}p}{p}\mathcal{D}_{u}\epsilon
^{\mu }.
\]

Eq. (\ref{XRef-Equation-22593926}) then takes the form
\begin{equation}
\mathcal{D}_{p}\mathcal{D}_{u}\epsilon _{\mu }-\frac{\nabla _{p}p}{p}\mathcal{D}_{u}\epsilon
_{\mu }=S_{\mu }^{\lambda }R_{\lambda \nu \alpha \beta }p^{\alpha
}u^{\beta }\epsilon ^{\nu }.%
\label{XRef-Equation-22594812}
\end{equation}

Before we proceed, let us introduce three new scalars.
\begin{definition}

Given a timelike unit vector $u^{\mu }$ and a space-like unit vector
${\hat{p}}^{\mu }$, define the {\bfseries {\itshape curvature twist}}
about $\hat{p}$ as the scalar\label{XRef-Definition-121263034}
\begin{equation}
\mathcal{Z}\equiv u^{\gamma }{\hat{p}}^{\beta }\epsilon ^{\rho \sigma
}R_{\gamma \rho \beta \sigma }.%
\label{XRef-Equation-113074020}
\end{equation}
\end{definition}

 An analysis of curvature twist will be given in Section \ref{XRef-Section-112692639}.
\begin{definition}

Given a plane electromagnetic wave with wave vector $p^{\mu }$ and
transverse polarization vector $\epsilon ^{\mu }$, define the {\bfseries
{\itshape polarization wiggle rate}} as the scalar\label{XRef-Definition-121073534}
\begin{equation}
\psi \equiv \epsilon ^{\gamma *}\epsilon _{\gamma \mu }\mathcal{D}_{u}\epsilon
^{\mu }.%
\label{XRef-Equation-112964427}
\end{equation}
\end{definition}

Let us stop for a moment and try to obtain an intuitive understanding
of this quantity. If $\epsilon ^{\mu }$ were a real, transverse
unit vector, $\mathcal{D}_{u}\epsilon ^{\mu }$ would be the transverse
component of its covariant time derivative, generating a rotation
of the polarization vector in the screen. Furthermore, its contraction
with the orthogonal vector $\epsilon ^{\gamma *}\epsilon _{\gamma
\mu }$ would yield the angular speed of this unit vector, as measured
in an inertial frame with\ \ geodesic $u^{\mu }$. However, since
$\epsilon ^{\mu }$ is a complex vector, rather than generating a
rotation, $\mathcal{D}_{u}\epsilon ^{\mu }$ generates a unitary
transformation of $\epsilon ^{\mu }$ in the screen. 

From its definition, $\psi $ appears to encode a rate of change
of the state of polarization in an inertial frame.
\begin{definition}

Given a plane electromagnetic wave with wave vector $p^{\mu }$ and
transverse polarization vector $\epsilon ^{\mu }$, define the {\bfseries
{\itshape mean helicity phase wiggle rate}} as the scalar
\begin{equation}
\chi \equiv -i \epsilon ^{\mu *}\mathcal{D}_{u}\epsilon _{\mu }.%
\label{XRef-Equation-103165626}
\end{equation}
\end{definition}

We will refer to $\psi $ and $\chi $ as the {\itshape polarization
wiggle scalars}.
\begin{remark}

While $\nabla _{u}\epsilon ^{\mu }=D \epsilon ^{\mu }/d\tau $ is
the time derivative of the transverse polarization vector, the polarization
wiggle scalars $\psi $ and $\chi $ are two independent transverse
projections of it quantifying proper time rates of change of the
state of polarization in an inertial frame. The two vectors used
in the definitions of the polarization wiggle scalars, $\epsilon
^{\mu \gamma }\epsilon _{\gamma }^{*}$ and $-i \epsilon ^{\mu *}$,
are complex unit vectors that form an orthonormal basis for a $\mathbb{C}^{2}$
representation of the vector space of transverse complex vectors.
Consequently, the two polarization wiggle scalars constitute a complete
representation of the time rate of change of the observed state
of polarization.\ \ 
\end{remark}

Let us contract eq. (\ref{XRef-Equation-22594812}) with the vectors
$\epsilon ^{\mu \gamma }\epsilon _{\gamma }^{*}$ and $-i \epsilon
^{\mu *}$. These vectors are by definition transverse. The resulting
transport equations for the polarization wiggle scalars are given
by the following lemma.
\begin{lemma}

For a given plane electromagnetic wave with a corresponding null
geodesic $p^{\mu }$ and rest frame momentum $p$, the transport equations
for the polarization wiggle scalars $\psi $ and $\chi $ along the
null geodesic are\label{XRef-Lemma-122013152}
\begin{gather}
\nabla _{p}\psi -\frac{\nabla _{p}p}{p}\psi =p \mathcal{Z}%
\label{XRef-Equation-112871124}
\\\nabla _{p}\chi -\frac{\nabla _{p}p}{p}\chi =-p \mathcal{V} \mathcal{Z},%
\label{XRef-Equation-1029172928}
\end{gather}

where $\mathcal{V}=i \epsilon ^{\mu \nu }P_{\mu \nu }$ is the circular
polarization degree, cf. Definition \ref{XRef-Definition-11674543}.
\end{lemma}
\begin{proof}

When contracting\ \ the left-hand side of eq. (\ref{XRef-Equation-22594812})
with $\epsilon ^{\mu \gamma }\epsilon _{\gamma }^{*}$, we may use
that\ \ $\mathcal{D}_{p}\epsilon ^{\mu \gamma }=0$ by\ \ Proposition
\ref{XRef-Proposition-112219738}. Since $\mathcal{D}_{p}\epsilon
^{\mu }=0$ by Proposition \ref{XRef-Proposition-112363540}, the
left-hand side of eq. (\ref{XRef-Equation-22594812}) transforms
to 
\[
\nabla _{p}\psi -\frac{\nabla _{p}p}{p}\psi .
\]

In Appendix \ref{XRef-Subsection-9108468} , we find that contracting
the right-hand side of eq. (\ref{XRef-Equation-22594812}) with $\epsilon
^{\mu \gamma }\epsilon _{\gamma }^{*}$ yields
\[
\epsilon _{\gamma }^{*}\epsilon ^{\gamma \mu }S_{\mu }^{\lambda
}R_{\lambda \nu \alpha \beta }p^{\alpha }u^{\beta }\epsilon ^{\nu
}=\epsilon _{\gamma }^{*}\epsilon ^{\mu \gamma }R_{\mu \lambda \alpha
\beta }p^{\alpha }u^{\beta }\epsilon ^{\lambda }=p \mathcal{Z}.
\]

Similarly, the second polarization wiggle equation, eq. (\ref{XRef-Equation-1029172928}),
can be proved by contracting eq. (\ref{XRef-Equation-22594812})
with $-i \epsilon ^{\mu *}$.
\end{proof}

It should be noted that eqs. (\ref{XRef-Equation-112871124}) and
(\ref{XRef-Equation-1029172928}) can be rewritten as
\begin{gather}
\nabla _{p}\left( \frac{\psi }{p}\right) =\mathcal{Z}%
\label{XRef-Equation-225164552}
\\\nabla _{p}\left( \frac{\chi }{p}\right) =-\mathcal{V}\mathcal{Z}.%
\label{XRef-Equation-1029173641}
\end{gather}

We will refer to eqs. (\ref{XRef-Equation-112871124})-(\ref{XRef-Equation-1029173641})
as the {\itshape polarization wiggle equations}. Their physical
interpretation will be given shortly, but first we will prove that
$\psi $ and $\chi $ are {\itshape frequency independent}, meaning
that they do not depend on the photon energy $p$ and thereby radiation
frequency. This is not apparent, since eqs. (\ref{XRef-Equation-112871124})
and (\ref{XRef-Equation-1029172928}) clearly depend on $p$.

Let $n^{\mu }\equiv p^{\mu }/p=u^{\mu }+{\hat{p}}^{\mu }$ be a scale
invariant null vector. Eq. (\ref{XRef-Equation-112871124}) can then
be rewritten as
\[
\nabla _{n}\psi -\frac{\nabla _{n}p}{p}\psi =\mathcal{Z}.
\]

The right-hand side of this equation is frequency independent. What
about the factor $\nabla _{n}p/p$? By using the definition of $p$
as $p\equiv -u_{\mu }p^{\mu }$, we find that $\nabla _{n}p/p$ can
be expressed in terms of the expansion tensor $\theta _{\mu \nu
}$ of the reference geodesic congruence as
\begin{equation}
\frac{1}{p}\nabla _{n}p=\frac{1}{p^{2}}\nabla _{p}p=-{\hat{p}}^{\mu
}{\hat{p}}^{\nu }{\theta }_{\mu \nu },%
\label{XRef-Equation-58123722}
\end{equation}

where
\begin{equation}
\theta _{\mu \nu }\equiv \nabla _{\left( \mu \right. }u_{\left.
\nu \right) }.%
\label{XRef-Equation-1227228}
\end{equation}

The expansion tensor can be decomposed as \cite{Poisson-2004}
\[
\theta _{\mu \nu }=\frac{1}{3}\theta  \gamma _{\mu \nu }+\sigma
_{\mu \nu },
\]

where $\sigma _{\mu \nu }$ is the symmetric, trace-free shear tensor
and $\theta \equiv \nabla _{\alpha }u^{\alpha }$ is the expansion
scalar of the tmelike geodesic. This proves the following result:
\begin{corollary}

Given a plane electromagnetic wave with wave vector $p^{\mu }$,
the polarization wiggle equations, eqs. (\ref{XRef-Equation-112871124})
and (\ref{XRef-Equation-1029172928}), can be written in the manifestly
frequency independent form\label{XRef-Corollary-12485247}
\begin{gather}
\nabla _{n}\psi +\kappa  \psi =\mathcal{Z}%
\label{XRef-Equation-1130185728}
\\\nabla _{n}\chi +\kappa  \chi =-\mathcal{V}\mathcal{Z},%
\label{XRef-Equation-1029174955}
\end{gather}

where $n^{\mu }\equiv p^{\mu }/p$ is the scale invariant null vector,
$\kappa $ is the projection of the expansion tensor of the timelike
reference geodesic congruence along the direction of propagation
defined as
\begin{equation}
\kappa \equiv {\hat{p}}^{\mu }{\hat{p}}^{\nu }\theta _{\mu \nu }=\frac{1}{3}\theta
+\Sigma ,%
\label{XRef-Equation-12272315}
\end{equation}

$\theta $ is the expansion and $\Sigma \equiv {\hat{p}}^{\mu }{\hat{p}}^{\nu
}\sigma _{\mu \nu }$ is the shear of the reference geodesic congruence
along the direction of propagation, $\hat{p}$. 
\end{corollary}
\begin{remark}

The transverse polarization vector no longer appears in the polarization
wiggle equations, eqs. (\ref{XRef-Equation-112871124}) and (\ref{XRef-Equation-1029172928}).
Eq. (\ref{XRef-Equation-112871124}) is completely independent of
the state of polarization, while eq. (\ref{XRef-Equation-1029172928})
only depends on the circular polarization degree, $\mathcal{V}$.
It implies that the polarization wiggle rate $\psi $ is independent
of\ \ the state of polarization, while the mean helicity wiggle
rate $\chi $ depends on the circular polarization degree.
\end{remark}
\begin{remark}

Eqs. (\ref{XRef-Equation-1130185728}) and (\ref{XRef-Equation-1029174955})
are frequency independent and do not depend on the the photon energy
$p$. Therefore, the evolution of the polarization wiggle scalars
along a given null geodesic are independent of radiation frequency.
As we will see later, this result allows the polarization wiggle
equations to be generalized to rays of composite electromagnetic
radiation with a geometric optics frequency spectrum and arbitrary
state of polarization.
\end{remark}

\subsection{The Polarization Wiggle Equations in Newman-Penrose
Formalism}

When expressed in terms of the complex null tetrad of Section \ref{XRef-Subsection-2269858},
the $\kappa $ scalar is
\[
\kappa =\frac{1}{\sqrt{2}}k^{\nu }( D-\Delta )  l_{\nu }.
\]

The terms on the right-hand side of this equation can be expressed
in terms of Newman-Penrose coefficients $\varepsilon $ and $\gamma
$ \cite{Newman-Penrose-1962,Stephani-2003}:
\begin{gather*}
k^{\nu }D l_{\nu }=-l^{\nu }D k_{\nu }=2\varepsilon -m_{\nu }^{+}D
m_{-}^{\nu }
\\k^{\nu }\Delta  l_{\nu }=2\gamma +m_{\nu }^{-}\Delta  m_{+}^{\nu
}.
\end{gather*}

Then
\[
k^{\nu }( D-\Delta )  l_{\nu }=2\varepsilon -2\gamma ^{*}-m_{\nu
}^{+}( D +\Delta ) m_{-}^{\nu }.
\]

The last term on the right-hand side of this equation vanishes by
eq. (\ref{XRef-Equation-102565226}), which gives $\kappa $ in terms
of Newman-Penrose spin coefficients:
\begin{equation}
\kappa =\sqrt{2}\left( \varepsilon -\gamma ^{*}\right) .
\end{equation}

The polarization wiggle equations, eq. (\ref{XRef-Equation-1130185728})
and (\ref{XRef-Equation-1029174955}) , then take the following form
in the Newman-Penrose formalism:
\begin{gather}
D \psi +\left( \varepsilon -\gamma ^{*}\right)  \psi =\frac{1}{\sqrt{2}}\mathcal{Z}
\\D \chi +\left( \varepsilon -\gamma ^{*}\right)  \chi =-\frac{1}{\sqrt{2}}\mathcal{V}\mathcal{Z}.
\end{gather}

\subsection{Physical Interpretation of the Polarization Wiggle Scalars}\label{XRef-Subsection-51152910}

Now, what quantities do the polarization wiggle scalars $\psi $
and $\chi $ represent, and can they be related to observable quantities?
Since both are complex quantities, we should consider the real and
imaginary parts separately. To gain an understanding of these quantities,
we will expand them in the null tetrad formalism introduced in Section
\ref{XRef-Subsection-2269858}. We will also use the tensor representation
of electromagnetic polarization introduced in Section \ref{XRef-Subsection-52815235}
above. More detail can be found in Appendix \ref{XRef-AppendixSection-430133648},
where we review the tensor representation of an arbitrary state
of electromagnetic polarization. 

Let us consider the imaginary part of $\psi $ first. From Proposition
\ref{XRef-Proposition-112219738}, we have $\mathcal{D}_{\alpha }\epsilon
_{\mu \nu }=0$. Then, from its definition in eq. (\ref{XRef-Equation-112964427})
follows that
\[
2i \operatorname{Im}[ \psi ] =-\mathcal{D}_{u}( \epsilon ^{\mu \nu
}i_{\mu \nu }^{\mathrm{pw}}) ,
\]

where $i_{\mu \nu }^{\mathrm{pw}}\equiv \epsilon _{\mu }\epsilon
_{\nu }^{*}$ is the relative coherency tensor for a plane electromagnetic
wave. Using the relationship between the relative coherency tensor
$i_{\mu \nu }$ and the polarization tensor $P_{\mu \nu }$ defined
in eq. (\ref{XRef-Equation-112318139}), we find that 
\[
\mathcal{D}_{u}( \epsilon ^{\mu \nu }i_{\mu \nu }^{\mathrm{pw}})
=\mathcal{D}_{u}( \epsilon ^{\mu \nu }P_{\left[ \mu \nu \right]
}) .
\]

 Then, by using the decomposition of the polarization tensor into
symmetric and antisymmetric parts of eq. (\ref{XRef-Equation-430131121}),
we obtain
\begin{equation}
\alpha \equiv \operatorname{Im}[ \psi ] =\frac{1}{2}\mathcal{D}_{u}\mathcal{V}=\frac{1}{2}\overset{\cdot
}{\mathcal{V}},%
\label{XRef-Equation-430141134}
\end{equation}

where $\mathcal{V}$ is the circular polarization degree, cf. Definition
\ref{XRef-Definition-11674543}. Overdot denotes differentiation
with respect to proper time. Thus, $\alpha \equiv \operatorname{Im}[
\psi ] $ is half the temporal rate of change of the circular polarization
degree of the radiation, as measured by an inertial observer. In
the following, we will refer to $\alpha $ as the {\itshape circular
polarization wiggle rate}.

Next, consider the real part of $\psi $. Define the {\itshape polarization
axis wiggle rate} as $\omega \equiv \operatorname{Re}[ \psi ] $.
The transverse polarization vector $\epsilon ^{\mu }$ can be expanded
in terms of the transverse vectors of the complex null tetrad as
\[
\epsilon ^{\mu }=\epsilon _{+}m_{+}^{\mu }+\epsilon _{-}m_{-}^{\mu
}.
\]

In this representation, we immediately recognize the helicity decomposition
of the polarization vector: $m_{+}^{\mu }$ and $m_{-}^{\mu }$ are
basis vectors with positive and negative helicity, respectively,
while $\epsilon _{\pm }$ are the helicity components of the polarization
vector. Since $\epsilon ^{\mu }$ is a complex unit vector, the helicity
components must satisfy $\epsilon _{+}^{*}\epsilon _{+}+\epsilon
_{-}^{*}\epsilon _{-}=1$. Using the expressions for the screen projector
and the screen rotator in the complex null tetrad basis, we find
that the polarization wiggle rate can be written
\begin{equation}
\psi =\frac{i}{\sqrt{2}}\left( \epsilon _{+}^{*}( D+\Delta ) \epsilon
_{+}-\epsilon _{-}^{*}( D+\Delta ) \epsilon _{-}+m_{\nu }^{-}( D+\Delta
) m_{+}^{\nu }\right) .%
\label{XRef-Equation-102572926}
\end{equation}

By eq. (\ref{XRef-Equation-102565226}), the last term on the right-hand
side of eq. (\ref{XRef-Equation-102572926}) vanishes. Furthermore,
the helicity components $\epsilon _{\pm }$ can be written as
\[
\epsilon _{+}=q_{+}e^{i \phi _{+}}, \epsilon _{-}=q_{-}e^{i \phi
_{-}},
\]

where $q_{\pm }$ are real amplitudes that must satisfy $q_{+}^{2}+q_{-}^{2}=1$
and $\phi _{\pm }$ are arbitrary real phases. $q_{\pm }^{2}$ are
relative intensities of the helicity components, and $\phi _{\pm
}$ are the helicity phases. We will use the term {\itshape helicity
phase wiggling }to refer to helicity phases that change with time,
and ${\overset{\cdot }{\phi }}_{\pm }$ will be denoted {\itshape
helicity phase wiggle rates}. The\ \ polarization wiggle rate now
takes the form
\begin{equation}
\psi =q_{-}^{2}{\overset{\cdot }{\phi }}_{-}-q_{+}^{2}{\overset{\cdot
}{\phi }}_{+}+\frac{i}{2}\nabla _{u}\left( q_{+}^{2}-q_{-}^{2}\right)
.%
\label{XRef-Equation-1029201849}
\end{equation}

We recognize $q_{+}^{2}-q_{-}^{2}$ as the circular polarization
degree $\mathcal{V}$ in the helicity basis \cite{Jackson-Electrodynamics},
and we see that $\operatorname{Im}[ \psi ]  =\overset{\cdot }{\mathcal{V}}/2
$, which is consistent with eq. (\ref{XRef-Equation-430141134})
above. Now, 
\begin{equation}
\omega \equiv \operatorname{Re}[ \psi ] =q_{-}^{2}{\overset{\cdot
}{\phi }}_{-}-q_{+}^{2}{\overset{\cdot }{\phi }}_{+}.%
\label{XRef-Equation-1029203639}
\end{equation}

Using the constraint $q_{+}^{2}+q_{-}^{2}=1$ allows us to express
the helicity intensities in terms of the the circular polarization
degree: 
\[
q_{\pm }^{2}=\frac{1}{2}\left( 1\pm \mathcal{V}\right) .
\]

By similar derivations, the mean helicity phase wiggle rate can
be expressed in terms of $q_{\pm }$ and $\phi _{\pm }$ as
\begin{equation}
\chi =q_{-}^{2}{\overset{\cdot }{\phi }}_{-}+q_{+}^{2}{\overset{\cdot
}{\phi }}_{+}.%
\label{XRef-Equation-1029203653}
\end{equation}

We see from eq. (\ref{XRef-Equation-1029203653}) that the mean helicity
phase wiggle rate $\chi $ can be interpreted as the mean of the
two helicity phase wiggle rates. Adding and subtracting eqs. (\ref{XRef-Equation-1029203639})
and (\ref{XRef-Equation-1029203653}) yields the the helicity phase
wiggle rates in terms of the wiggle rates $\omega $ and $\chi $:
\begin{gather}
 q_{-}^{2}{\overset{\cdot }{\phi }}_{-}=\frac{1}{2}\left( \chi +\omega
\right) %
\label{XRef-Equation-1029205237}
\\q_{+}^{2}{\overset{\cdot }{\phi }}_{+}=\frac{1}{2}\left( \chi
-\omega \right) .%
\label{XRef-Equation-1029205312}
\end{gather}

Next, let $\Omega $ denote the angle of the polarization axis in
the frame of observation. It is undefined for purely circular polarization
($\mathcal{V}=\pm 1$), but it is an observable for linear and elliptical
polarization. It is related to the Stokes parameters as follows
\cite{Born-Wolf-Optics}:
\begin{equation}
\tan  2\Omega =\frac{U}{Q}=\frac{\mathcal{U}}{\mathcal{Q}}.%
\label{XRef-Equation-10237735}
\end{equation}

Furthermore, the ratio on the right-hand side of eq. (\ref{XRef-Equation-10237735})
can be expressed in terms of the helicity phase difference \cite{Jackson-Electrodynamics}:
\[
\frac{U}{Q}=\tan ( \phi _{-}-\phi _{+}) .
\]

This lets us relate the helicity phase difference to the polarization
angle $\Omega $:
\[
\phi _{-}-\phi _{+}=2\Omega .
\]

The rotation rate of the polarization axis, $\overset{\cdot }{\Omega
}$, can now be expressed in terms of the polarization wiggle rates
$\omega $ and $\chi $:
\begin{equation}
\overset{\cdot }{\Omega }=\frac{1}{2}\left( {\overset{\cdot }{\phi
}}_{-}-{\overset{\cdot }{\phi }}_{+}\right) =\frac{\omega +\mathcal{V}
\chi }{1-\mathcal{V}^{2}}.%
\label{XRef-Equation-11721125}
\end{equation}

We see that $\omega =\overset{\cdot }{\Omega }$ for linear polarization
($\mathcal{V}=0$). 

\subsection{Polarization Wiggling Observables}

We have identified quantities that quantify effects of gravity on
the polarization of a plane electromagnetic wave. The following
quantities are observables of rates of change of the state of polarization
in the frame of observation: 
\begin{itemize}
\item ${\overset{\cdot }{\phi }}_{\pm }$: The {\itshape helicity
phase wiggle rates. }They quantify the rate of change of the helicity
phases
\item $\overset{\cdot }{\mathcal{V}}$: The rate of change of the
{\itshape circular polarization degree }
\item $\overset{\cdot }{\Omega }$: The angular rotation rate of
the polarization axis
\end{itemize}

We will refer to these observables as {\itshape polarization wiggling
observables}. The Stokes parameters and the helicity phases $\phi
_{\pm }$ are {\itshape polarization state observables}.
\begin{remark}

If we assume that the radiation is emitted with constant polarization,
the polarization state observables must be constant in the frame
of emission. This implies that all the polarization wiggling observables
are zero in the frame of emission in this case.
\end{remark}

The two polarization wiggle scalars $\psi $ and $\chi $ can be expressed
in terms of these observables. We can summarize these results as
follows:
\begin{lemma}

For a given plane electromagnetic wave, the polarization wiggle
rate $\psi $ encodes two scalar quantities; the {\bfseries polarization
axis wiggle rate} $\omega $ and the {\bfseries circular polarization
wiggle rate} $\alpha $. Expressed in terms of the polarization wiggling
observables, they are\label{XRef-Lemma-1227028}
\begin{gather}
\omega \equiv \operatorname{Re}[ \psi ] =\frac{1}{2}\left( \left(
1-\mathcal{V}\right) {\overset{\cdot }{\phi }}_{-}-\left( 1+\mathcal{V}\right)
{\overset{\cdot }{\phi }}_{+}\right) =\overset{\cdot }{\Omega }-\frac{\mathcal{V}}{2}\left(
{\overset{\cdot }{\phi }}_{+}+{\overset{\cdot }{\phi }}_{-}\right)
\label{XRef-Equation-111165030}
\\\alpha \equiv \operatorname{Im}[ \psi ] =\frac{1}{2}\overset{\cdot
}{\mathcal{V}},
\end{gather}

where $\overset{\cdot }{\Omega }$ is the rotation rate of the polarization
axis and $\mathcal{V}$ is circular polarization degree, both measured
in an inertial frame. $\phi _{+}$ and $\phi _{-}$ are the phases
of the helicity components. The mean helicity phase wiggle rate
$\chi $ is
\begin{equation}
\chi =\frac{1}{2}\left( \left( 1-\mathcal{V}\right) {\overset{\cdot
}{\phi }}_{-}+\left( 1+\mathcal{V}\right) {\overset{\cdot }{\phi
}}_{+}\right) .
\end{equation}

The inverse relationships are
\begin{gather}
\overset{\cdot }{\Omega }=\frac{\omega +\mathcal{V} \chi }{1-\mathcal{V}^{2}}
\\{\overset{\cdot }{\phi }}_{+}=\frac{\chi -\omega }{1+\mathcal{V}}
\\{\overset{\cdot }{\phi }}_{-}=\frac{\chi +\omega }{1-\mathcal{V}}.
\end{gather}
\end{lemma}

The polarization wiggle scalars $\psi $ and $\chi $ are quantities
with simple transport equations given by Lemma \ref{XRef-Lemma-122013152}.
They are defined for any state of polarization. These scalars are
not measurable quantities, but they are linear combinations of such
quantities. Through the linear relationships stated in Lemma \ref{XRef-Lemma-122013152}
to the polarization wiggling observables, it is straight forward
to transform between polarization wiggle scalars and polarization
wiggling observables.

\subsection{Circular Polarization Wiggle Rate}

Taking the imaginary part of the polarization wiggle equation, eq.
(\ref{XRef-Equation-225164552}), gives a transport equation for
the circular polarization wiggle rate. By Lemma \ref{XRef-Lemma-1227028},
the left-hand side of the equation can be expressed in terms of
the circular polarization wiggle rate $\alpha $. The right-hand
side of eq. (\ref{XRef-Equation-225164552}) is real and does not
contribute to $\alpha $. This yields the {\itshape circular polarization
wiggle equation, }a transport equation for the circular polarization
wiggle rate $\alpha $:
\begin{lemma}

For a given plane electromagnetic wave with wave vector $p^{\mu
}$ and rest frame momentum $p$, the quantity $\alpha /p$ is conserved
along the null geodesic:
\begin{equation}
\nabla _{p}\left( \frac{\alpha }{p}\right) =0,%
\label{XRef-Equation-124101546}
\end{equation}

where $\alpha $ is the circular polarization wiggle rate.
\end{lemma}

The circular polarization wiggle equation can easily be solved,
yielding the following result.
\begin{theorem}

For a given plane electromagnetic wave with wave vector $p^{\mu
}$ and rest frame momentum $p$, the circular polarization wiggle
rate $\alpha $ evolves as\label{XRef-Theorem-12481042}
\begin{equation}
\alpha ( \lambda ) =\frac{p( \lambda ) }{p( \lambda _{*}) }\left.
\alpha (\lambda _{*}\right) ,%
\label{XRef-Equation-12481212}
\end{equation}

where $\lambda $ is the affine parameter along the null geodesic
and $\lambda _{*}$ labels the emission event.
\end{theorem}
\begin{remark}

From Theorem \ref{XRef-Theorem-12481042} follows that, if the circular
polarization degree is constant in the inertial frame of emission
($\overset{\cdot }{\mathcal{V}}(\lambda _{*})=0$), it will remain
constant ($\overset{\cdot }{\mathcal{V}}( \lambda ) =0$) at any
point along the null geodesic. \label{XRef-Remark-11127112}
\end{remark}

Eq. (\ref{XRef-Equation-12481212}) can be expressed in terms of
the circular polarization degree $\mathcal{V}$ and the redshift
$z( \lambda _{*}) $ of the source:
\begin{equation}
\overset{\cdot }{\mathcal{V}}( \lambda ) ={\left( 1+z( \lambda _{*})
\right) }^{-1} \overset{\cdot }{\mathcal{V}}\left(\lambda _{*}\right)
.%
\label{XRef-Equation-5110724}
\end{equation}

Thus, a non-zero circular polarization wiggle rate of radiation
is redshifted with the same factor as radiation frequencies. 

\subsection{Polarization Axis Wiggle Rate}

From eq. (\ref{XRef-Equation-111165030}) of Lemma \ref{XRef-Lemma-1227028},
we see that the polarization axis wiggle rate $\omega $ is related
to the rotation rate of the polarization axis, $\overset{\cdot }{\Omega
}$. For linear polarization ($\mathcal{V}=0$), the relationship
is an identity; $\omega |_{\mathcal{V}=0}=\overset{\cdot }{\Omega
}.$ For elliptical polarization ($\mathcal{V}\neq 0$), $\omega $
also gets contributions from the two helicity phase wiggle rates
${\overset{\cdot }{\phi }}_{\pm }$. For positive circular polarization
($\mathcal{V}=1$), $\omega |_{\mathcal{V}=1}=-{\overset{\cdot }{\phi
}}_{+}$, while for negative circular polarization ($\mathcal{V}=-1$),
$\omega |_{\mathcal{V}=-1}={\overset{\cdot }{\phi }}_{-}$.

Taking the real part of the polarization wiggle equation, eq. (\ref{XRef-Equation-112871124})
of Lemma \ref{XRef-Lemma-122013152} and the equivalent eq. (\ref{XRef-Equation-1029172928}),
yields the {\itshape polarization axis wiggle equation}; a transport
equation for the {\itshape polarization axis wiggle rate} $\omega
$:
\begin{lemma}

For a given plane electromagnetic wave with wave vector $p^{\mu
}$ and rest frame momentum $p$, the transport equation for the polarization
axis wiggle rate $\omega $ is\label{XRef-Lemma-12147341}
\begin{equation}
\nabla _{p}\omega -\frac{\nabla _{p}p}{p}\omega =p \mathcal{Z},%
\label{XRef-Equation-1249368}
\end{equation}

equivalent to the equation
\begin{equation}
\nabla _{p}\left( \frac{\omega }{p}\right) =\mathcal{Z}.
\end{equation}
\end{lemma}

Polarization axis wiggling is induced by a non-zero curvature twist
along the radiation geodesic. The polarization axis wiggle equation
can be solved by integration along the null geodesic, yielding the
following result.
\begin{theorem}

For a given plane electromagnetic wave with wave vector $p^{\mu
}$ and rest frame momentum $p$, the general line-of-sight solution
to the polarization axis wiggle equation is\label{XRef-Theorem-124922}
\begin{equation}
\omega ( \lambda ) =\frac{p( \lambda ) }{p( \lambda _{*}) }\omega
( \lambda _{*}) +p( \lambda ) \text{}\operatorname*{\int }\limits_{\lambda
_{*}}^{\lambda }d\lambda ^{\prime }\ \ \mathcal{Z}( \lambda ^{\prime
}) ,%
\label{XRef-Equation-1249239}
\end{equation}

where $\lambda $ is the affine parameter along the null geodesic
and $\lambda _{*}$ labels the emission event.
\end{theorem}

If the emitted radiation has constant polarization in the inertial
frame of emission, $\omega ( \lambda _{*}) =0$, and the wiggle rate
simplifies to
\begin{equation}
\omega ( \lambda ) =p( \lambda ) \text{}\operatorname*{\int }\limits_{\lambda
_{*}}^{\lambda }d\lambda ^{\prime }\ \ \mathcal{Z}( \lambda ^{\prime
}) .%
\label{XRef-Equation-22775842}
\end{equation}

The polarization axis wiggle rate $\omega $ is independent of intrinsic
properties of the radiation, such as frequency and state of polarization.
Eqs. (\ref{XRef-Equation-1249239}) and (\ref{XRef-Equation-22775842})
give covariant expressions for the polarization axis wiggle rate.
These expressions can be plugged into Lemma \ref{XRef-Lemma-1227028}
to obtain expressions for the polarization change observables. 

Within the geometric optics limit, Theorem \ref{XRef-Theorem-12481042}
and Theorem \ref{XRef-Theorem-124922} are completely general; valid
for any 4-dimensional spacetime geometry.

\subsection{Mean Helicity Phase Wiggle Rate}

The transport equation for the mean helicity phase wiggle rate $\chi
$ is given by eq. (\ref{XRef-Equation-1029172928}) of Lemma \ref{XRef-Lemma-122013152}.
It can be solved by integration along the null geodesic, yielding
the following result.
\begin{theorem}

For a given plane electromagnetic wave with wave vector $p^{\mu
}$ and rest frame momentum $p$, the general line-of-sight solution
to the transport equation for the mean helicity phase wiggle rate
is
\begin{equation}
\chi ( \lambda ) =\frac{p( \lambda ) }{p( \lambda _{*}) }\chi (
\lambda _{*}) -p( \lambda ) \text{}\mathcal{V}\operatorname*{\int
}\limits_{\lambda _{*}}^{\lambda }d\lambda ^{\prime }\ \ \mathcal{Z}(
\lambda ^{\prime }) ,%
\label{XRef-Equation-112225647}
\end{equation}

where $\lambda $ is the affine parameter along the null geodesic
and $\lambda _{*}$ labels the emission event.
\end{theorem}
\begin{proof}

Eq. (\ref{XRef-Equation-112225647}) follows by using that $d\mathcal{V}/d\lambda
=0$ by eq. (\ref{XRef-Equation-11222535}) of Proposition \ref{XRef-Proposition-112225615}.
\end{proof}

If the emitted radiation has constant polarization in the inertial
frame of emission, $\chi ( \lambda _{*}) =0$, and the mean hellicity
wiggle rate simplifies to
\begin{equation}
\chi ( \lambda ) =-p( \lambda ) \mathcal{V}\text{}\operatorname*{\int
}\limits_{\lambda _{*}}^{\lambda }d\lambda ^{\prime }\ \ \mathcal{Z}(
\lambda ^{\prime }) .%
\label{XRef-Equation-11363322}
\end{equation}

\subsection{Radiation Emitted with Constant Polarization}

Let us consider the special case for electromagnetic radiation emitted
with constant polarization. In this case, $\omega ( \lambda _{*})
=\chi ( \lambda _{*}) =0$. From eqs. (\ref{XRef-Equation-22775842})
and (\ref{XRef-Equation-11363322}) follows that
\begin{equation}
\chi ( \lambda ) =-\omega ( \lambda ) \mathcal{V}.%
\label{XRef-Equation-11364957}
\end{equation}

The relationship between the polarization wiggling observables and
the polarization axis wiggle rate simplifies considerably in this
case. The following corollary follows directly from Lemma \ref{XRef-Lemma-1227028}
by the use of eq. (\ref{XRef-Equation-11364957}).
\begin{corollary}

For a plane electromagnetic wave emitted with constant linear or
elliptical polarization ($\mathcal{V}\neq \pm 1$), the polarization
wiggling observables can be expressed in terms of the polarization
axis wiggle rate $\omega $ as follows:\label{XRef-Corollary-11721050}
\begin{gather}
\overset{\cdot }{\Omega }=\omega 
\\{\overset{\cdot }{\phi }}_{+}=-\omega 
\\{\overset{\cdot }{\phi }}_{-}=\omega 
\\\alpha =0.
\end{gather}
\end{corollary}

In this case, the helicity phase wiggle rates can be expressed in
terms of the polarization axis wiggle rate:
\begin{gather*}
{\overset{\cdot }{\phi }}_{+}=-\overset{\cdot }{\Omega }
\\{\overset{\cdot }{\phi }}_{-}=\overset{\cdot }{\Omega }.
\end{gather*}

\subsection{Relating the Polarization Axis Wiggle Rate to the Gravitational
Faraday Effect}

In this section, we will demonstrate how the polarization axis wiggle
rate $\omega $ can be related to the rotation rate of the Gravitational
Faraday Effect. We will relate our results to the results of Kopeikin
and Mashoon referenced in Section \ref{XRef-Subsection-66111547}.
Linear gravity on a Minkowski background is assumed, and we will
assume that the radiation can be represented by a linearly polarized
plane wave. For reasons of simplicity, the demonstration will be
limited to the case of an observer and emitter with infinitesimal
separation.

Our derivation of the transport equation for the polarization wiggle
rate $\psi $ starts with eq. (\ref{XRef-Equation-6692314}). A common
way of proving this equation is to define an infinitesimal loop
as a quadrilateral spanned by the two vectors $X$ and $Y$ and parallel
transport a third vector, $Z$, around that loop \cite{Misner}. In
this paper, we let $Z$ be the transverse polarization vector $\epsilon
^{\mu }$, while $X$ and $Y$ were the null and timelike geodesics
$p^{\mu }$ and $u^{\mu }$, respectively. Assuming an inertial emitter
and observer with infintesimal separation, a single infinitesimal
loop can then be defined as a spacetime quadrilateral defined by
two emission events $\mathfrak{E}_{1}$ and $\mathfrak{E}_{2}$ and
two observation events $\mathfrak{O}_{1}$ and $\mathfrak{O}_{2}$
with four infinitesimal geodesic segments linking them: $\mathfrak{U}_{E}:\mathfrak{E}_{1}\rightarrow
\mathfrak{E}_{2}$, $\mathfrak{U}_{O}: \mathfrak{O}_{1}\rightarrow
\mathfrak{O}_{2}$, $\mathfrak{P}_{1}:\mathfrak{E}_{1}\rightarrow
\mathfrak{O}_{1}$ and $\mathfrak{P}_{2}:\mathfrak{E}_{2}\rightarrow
\mathfrak{O}_{2}$. Let $\Omega $ denote the observed angle of the
polarization axis in a given polarization basis in the frame of
observation. Assuming that the polarization basis is stationary
in the frame of the observer, the angular change $\Delta \Omega
$ in orientation of the polarization axis between the two measurement
events is an observable. 

Next, let us assume that the value of $\epsilon ^{\mu }$ at the
first measurement event, $\epsilon ^{\mu }( \mathfrak{O}_{1}) $,
is parallel transported along $\mathfrak{U}_{O}$ to be compared
with the value at the second measurement, $\epsilon ^{\mu }( \mathfrak{O}_{2})
$. Assume, furthermore, that the emitted radiation has constant
polarization, which means that the second emitted value, $\epsilon
^{\mu }( \mathfrak{E}_{2}) $, is equal to the first emitted value,
$\epsilon ^{\mu }( \mathfrak{E}_{1}) $, parallel transported along
$\mathfrak{U}_{E}$. By integrating the transport equation of the
transverse polarization vector $\epsilon ^{\mu }$, eq. (\ref{XRef-Equation-1122174843}),
along the null geodesic segments $\mathfrak{P}_{1}$ and $\mathfrak{P}_{2}$,
we obtain the following expression for $\Delta \Omega $, using linear
gravity with a metric perturbation $h_{\mu \nu }$ to a Minkowski
background:
\begin{equation}
\Delta \Omega =N^{\left( p\right) }( \mathfrak{P}_{2}) -N^{\left(
p\right) }( \mathfrak{P}_{1}) -\left( N^{\left( u\right) }( \mathfrak{U}_{O})
-N^{\left( u\right) }( \mathfrak{U}_{E}) \right) .%
\label{XRef-Equation-66111845}
\end{equation}

$N^{(q)}( \mathfrak{S}) $ is the segment integral
\[
N^{\left( q\right) }( \mathfrak{S}) \equiv \operatorname*{\int }\limits_{\mathfrak{S}}ds
\nu ^{\left( q\right) }
\]

for a geodesic segment $\mathfrak{S}$ with geodesic vector $q$ and
affine parameter $s$, while
\begin{equation}
\nu ^{\left( q\right) }\equiv \frac{1}{2}\epsilon ^{\alpha \beta
}q^{\lambda }h_{\lambda [ \alpha ,\beta ] }.%
\label{XRef-Equation-66105741}
\end{equation}

When comparing eqs. (\ref{XRef-Equation-177445}) and (\ref{XRef-Equation-66105741}),
we see that, given a stationary metric, the rotation rate $\nu $
of eq. (\ref{XRef-Equation-177445}) is related to $\nu ^{(p)}$ as
$\nu =\nu ^{(p)}/p^{2}$, where $p$ is the scalar photon momentum.
The factor $1/p^{2}$ is due to different parametrization of the
geodesic curve. Thus, the results of Kopeikin and Mashoon referenced
in Section \ref{XRef-Subsection-66111547} correspond to the integral
$N^{(p)}( \mathfrak{P}_{a}) $. 

On the other hand, if we let $\Delta \tau $ denote the proper time
difference between the two observation events, the polarization
axis wiggle rate $\omega $ can be expressed in terms of $\Delta
\Omega $ as follows:
\[
\omega =\frac{\Delta \Omega }{\Delta \tau }.
\]

This gives an explicit relationship between the polarization axis
wiggle rate $\omega $ and the previous results of Kopeikin and Mashoon
\cite{Kopeikin-2001}.

It was remarked in Section \ref{XRef-Subsection-66111547} that we
should expect contributions from the motions of the frames of observation
and emission to an observable quantity such as $\Delta \Omega $.
This is now evident in eq. (\ref{XRef-Equation-66111845}), as the
term $-N^{(u)}( \mathfrak{U}_{O}) $ is the contribution to $\Delta
\Omega $ from motion of the frame of observation, while $N^{(u)}(
\mathfrak{U}_{E}) $ is the contribution to $\Delta \Omega $ from
motion of\ \ the frame of emission.

Notice that this demonstration of the relationship between the polarization
wiggle rate and the rotation rate of the Gravitational Faraday Effect
was for reasons of simplicity limited to the case of an observer
and emitter with infinitesimal separation, which can be represented
by a single quadrilateral of infinitesimal extent. 

\section{The Effect of Gravity on the Polarization of Composite
Electromagnetic Radiation}\label{XRef-Section-52794743}

So far we have limited our attention to plane electromagnetic waves,
representing coherent, monochromatic radiation, and how its polarization
is affected by gravity. Let us take a step further and consider
how gravity affects the state of polarization of a ray of composite
electromagnetic radiation with a geometric optics frequency spectrum
and in an arbitrary state of polarization. It is assumed that the
radiation is emitted with a constant frequency spectrum. We will
model the ray as an ensemble of photons, represented as plane wave
components. The ensemble definition is given in Section \ref{XRef-Subsection-52815235}.
We will assume that geometric optics applies to each plane wave
component of the ensemble, which allows us to draw on the results
of previous sections. All components of the ensemble therefore share
the same null geodesic and have the same scale invariant null tangent
vector field $n^{\mu }\equiv p^{\mu }/p=u^{\mu }+{\hat{p}}^{\mu
}$.

\subsection{Intensity Transport}

In geometric optics, the transport equation for the magnitude $a$
of an electromagnetic plane wave with 4-momentum $p^{\mu }$ along
its null geodesic is \cite{Misner} 
\begin{equation}
p^{\beta }\nabla _{\beta }a+\frac{1}{2}\left( \nabla _{\beta }p^{\beta
}\right) a=0.%
\label{XRef-Equation-58123044}
\end{equation}

Let $n^{\mu }\equiv p^{\mu }/p=u^{\mu }+{\hat{p}}^{\mu }$ be the
scale invariant null vector introduced in Section \ref{XRef-Subsection-58122719}
above. 
\begin{proposition}

The transport equation for the scalar intensity $I^{\mathrm{pw}}=a^{2}p^{2}$
of a plane electromagnetic wave is\label{XRef-Proposition-1267346}
\begin{equation}
\nabla _{n}I^{\mathrm{pw}}=-\beta  I^{\mathrm{pw}},%
\label{XRef-Equation-1267423}
\end{equation}

where
\begin{equation}
\beta \equiv \frac{1}{p}\nabla _{\alpha }p^{\alpha }-\frac{2}{p}\nabla
_{n}p.%
\label{XRef-Equation-1267030}
\end{equation}
\end{proposition}
\begin{proof}

The result follows from using eq. (\ref{XRef-Equation-58123044})
and the $p^{\mu }$ decomposition $p^{\mu }=p n^{\mu }$ of eq. (\ref{XRef-Equation-62782635}).
\end{proof}

The term $\nabla _{\alpha }p^{\alpha }$ on the right-hand side of
eq. (\ref{XRef-Equation-1267030}) is the expansion of the null congruence
$p^{\mu }$, which measures the relative expansion rate of a small
cross sectional area of the beam along its null geodesic \cite{Poisson-2004}.
As expected, the intensity decreases if the beam expands. The second
term on the right-hand side of eq. (\ref{XRef-Equation-1267030})
measures the relative rate of change of the photon energy along
the beam. As expected, the intensity decreases if the photon energy
decreases along the geodesic.

We notice that the scalar $\beta $ is independent of the magnitude
$a$. By using eq. (\ref{XRef-Equation-58123722}) above, we find
that
\[
\beta =\nabla _{\alpha } n^{\alpha }+\kappa ,
\]

where $\kappa $ is defined by eq. (\ref{XRef-Equation-12272315}).
$\kappa \equiv {\hat{p}}^{\mu }{\hat{p}}^{\nu }\theta _{\mu \nu
}$ is the projection of the expansion tensor $\theta _{\mu \nu }$
of the timelike reference geodesic congruence along the direction
of propagation. Thus, $\beta $ only depends on ${\hat{p}}^{\mu }$
and $u^{\mu }$ and is independent of both photon energy $p$ and
magnitude $a$. 
\begin{lemma}

Given an ensemble of photons with shared null geodesic, each represented
by a plane wave component with intensity $I_{(j)}^{\mathrm{pw}}$,
the transport equation for the total intensity $I=\sum _{j}I_{(j)}^{\mathrm{pw}}$
of the ensemble along the geodesic is\label{XRef-Lemma-126173534}
\begin{equation}
 \nabla _{n}I=-\beta  I.%
\label{XRef-Equation-12617302}
\end{equation}
\end{lemma}
\begin{proof}

The ensemble intensity is given by eq. (\ref{XRef-Equation-126172427})
as a sum of the component intensities. Eq. (\ref{XRef-Equation-12617302})
then follows from Proposition \ref{XRef-Proposition-1267346}, because
$\beta $ is independent of photon energy $p$ and magnitude $a$.
\end{proof}

Lemma \ref{XRef-Lemma-126173534} generalizes the transport equation
for the intensity of a plane electromagnetic wave, eq. (\ref{XRef-Equation-1267423})
of Proposition \ref{XRef-Proposition-1267346}, to a ray of composite
electromagnetic radiation with a geometric optics frequency spectrum
and arbitrary state of polarization. 

\subsection{Transport of the Polarization Tensor}

As defined in Section \ref{XRef-Subsection-52815235}, the state
of polarization of a ray of electromagnetic radiation can be represented
in a covariant way by the polarization tensor $P_{\mu \nu }$ defined
in eq. (\ref{XRef-Equation-112318139}) of Definition \ref{XRef-Definition-12771128}.\ \ Again,
we will represent the ray as an ensemble of plane wave components.
We will now derive the transport equation for the polarization tensor
$P_{\mu \nu }$ of this ensemble. 

Define the intensity fraction $i_{j}\equiv I_{(j)}^{\mathrm{pw}}/I$
of plane wave component $j$ of the ensemble. 
\begin{proposition}

Given an ensemble of photons with shared null geodesic, the intensity
fraction $i_{(j)}\equiv I_{(j)}^{\mathrm{pw}}/I$ of plane wave component
$j$ is conserved along the null geodesic and satisfies the transport
equation\label{XRef-Proposition-12763324}
\begin{equation}
\nabla _{n}i_{\left( j\right) }=0.%
\label{XRef-Equation-126174517}
\end{equation}
\end{proposition}
\begin{proof}

Eq. (\ref{XRef-Equation-126174517}) follows from evaluating $\nabla
_{n}i_{(j)}$ and applying Proposition \ref{XRef-Proposition-1267346}
and Lemma \ref{XRef-Lemma-126173534}.
\end{proof}
\begin{proposition}

Given an ensemble of photons with shared null geodesic, each represented
by a plane wave component with relative coherency tensor $i_{\mu
\nu }^{\mathrm{pw}}$, the transport equation of the relative coherency
tensor $i_{\mu \nu }\equiv I_{\mu \nu }/I$ is\label{XRef-Proposition-12765814}
\begin{equation}
\mathcal{D}_{n}i_{\mu \nu }=0.
\end{equation}
\end{proposition}
\begin{proof}

This follows from applying\ \ Proposition \ref{XRef-Proposition-112363540},
Proposition \ref{XRef-Proposition-12764641} and Proposition \ref{XRef-Proposition-12763324}
to the definition $i_{\mu \nu }\equiv I_{\mu \nu }/I=\sum _{j}i_{(j)}i_{\mu
\nu }^{\mathrm{pw}}$.
\end{proof}
\begin{proposition}

The coherency tensor $I_{\mu \nu }$ satisfies the transport equation
\[
\mathcal{D}_{n}I_{\mu \nu }+\beta  I_{\mu \nu }=0.
\]
\end{proposition}
\begin{proof}

This follows from applying Lemma \ref{XRef-Lemma-126173534} and
Proposition \ref{XRef-Proposition-12765814} to the relationship
$I_{\mu \nu }=I i_{\mu \nu }$.
\end{proof}

We are now finally in a position to derive the transport equation
for the polarization tensor $P_{\mu \nu }$. The polarization tensor
is defined in\ \ Definition \ref{XRef-Definition-12771128} as $P_{\mu
\nu }\equiv i_{\mu \nu }-\frac{1}{2}S_{\mu \nu }$. 
\begin{lemma}

Given a ray of electromagnetic radiation, its polarization tensor
$P_{\mu \nu }$ satisfies the transport equation\label{XRef-Lemma-1271865}
\begin{equation}
\mathcal{D}_{n}P_{\mu \nu }=0.%
\label{XRef-Equation-12771725}
\end{equation}
\end{lemma}
\begin{proof}

By representing the ray as an ensemble of photons with a shared
null geodesic, eq. (\ref{XRef-Equation-12771725}) follows by applying
Proposition \ref{XRef-Proposition-11236541} and Proposition \ref{XRef-Proposition-12765814}
to the definition of the polarization tensor.
\end{proof}

The polarization degree $\mathcal{P}$ is the degree to which the
radiation is polarized. From its definition in Definition \ref{XRef-Definition-12773823},
it can be expressed in terms of the polarization tensor as
\begin{equation}
\mathcal{P}^{2}=2P^{\mu \nu }P_{\nu \mu }=2P^{\mu \nu }P_{\mu \nu
}^{*}.%
\label{XRef-Equation-7375614}
\end{equation}

The polarization degree is an observable scalar, and its transport
equation is given by the following theorem.
\begin{theorem}

Given a ray of composite electromagnetic radiation with a geometric
optics frequency spectrum and arbitrary state of polarization, its
polarization degree $\mathcal{P}$ is constant along the null geodesic:\label{XRef-Theorem-127181215}
\begin{equation}
\nabla _{n}\mathcal{P}^{2}=0.%
\label{XRef-Equation-12718937}
\end{equation}
\end{theorem}
\begin{proof}

The ray can be represented as an ensemble of photons with a shared
null geodesic. Then, evaluating $\mathcal{D}_{n}\mathcal{P}^{2}$
and applying Lemma \ref{XRef-Lemma-1271865} to the definition of
$\mathcal{P}$ in eq. (\ref{XRef-Equation-7375614}) yields eq. (\ref{XRef-Equation-12718937}).
\end{proof}

Theorem \ref{XRef-Theorem-127181215} states that the polarization
degree $ \mathcal{P}$ is conserved along the null geodesic for any
electromagnetic radiation within the geometric optics limit. Notice
that the derivations of this section were done without any assumption
being made with respect to the spacetime geometry. Therefore, Lemma
\ref{XRef-Lemma-1271865} and Theorem \ref{XRef-Theorem-127181215}\ \ are
valid for any four-dimensional spacetime.

\subsection{Polarization Wiggling of Composite Radiation}

The aim of this section is to generalize the phenomenon of gravity-induced
polarization wiggling to composite electromagnetic radiation with
a geometric optics frequency spectrum and arbitrary state of polarization.
In order to do so, we need to isolate effects of gravity from other
effects that also may induce changes to the observed state of polarization.
For composite electromagnetic radiation, there are several aspects
related to the emission process that may cause the state of polarization
to change and therefore need to be controlled.

First, an emitter of electromagnetic radiation may in general emit
radiation with a frequency spectrum that varies with time. This
may cause the observed state of polarization to change. Since our
main objective is to study gravity-induced changes to the observed
state of polarization, we will assume that the radiation is emitted
with a constant frequency spectrum. 

The second concern is how to represent the emitter. In the above
sections, we have already chosen to represent a ray of composite
electromagnetic radiation as an ensemble of plane waves. When polarization
wiggling of a single plane wave was studied in Section \ref{XRef-Section-331204723},
we assumed the presence of a single emitter emitting a plane wave.
To generalize this to emission of composite radiation, we will assume
an ensemble of plane wave emitters. It should be stated that this
way of representing the emission process does not pretend to represent
the physical reality of how electromagnetic radiation is emitted,
just the desired result (an ensemble of plane waves).

Finally, to make sure that we study gravity-induced polarization
wiggling in isolation from other effects inducing changes to the
observed state of polarization, we will assume that each plane wave
emitter in the emitter ensemble emits plane waves with constant
polarization. Since the frequency spectrum is constant, this implies
that the radiation is emitted with a constant state of polarization.

The definition of the polarization wiggle rate given in Definition
\ref{XRef-Definition-121073534} is restricted to a single plane
electromagnetic wave. This definition needs to be generalized to
be relevant in the context of composite electromagnetic radiation
with arbitrary state of polarization.

The frequency independent form of the polarization wiggle equation,
eq. (\ref{XRef-Equation-1130185728}), tells us that the polarization
axis of a plane electromagnetic wave wiggles with a rate that is
independent of the photon energy. Intuitively, this indicates that
the polarization axis of a ray of partially coherent radiation,
defined as the major principal axis of the coherency tensor $I_{\mu
\nu }$, wiggles with the same rate as each of its plane wave components.
By Proposition \ref{XRef-Proposition-121095312}, a vector aligned
with the major principal axis of the coherency tensor is transverse.
This indicates that for a ray of composite radiation, a unit vector
along the major principal axis of the coherency tensor $I_{\mu \nu
}$ can be used as the transverse polarization vector $\epsilon ^{\mu
}$ in Definition \ref{XRef-Definition-121073534}. We will now prove
this.
\begin{definition}

For a given ray of composite electromagnetic radiation with a geometric
optics frequency spectrum and arbitrary state of polarization, define
the {\bfseries {\itshape transverse polarization vector}} $\epsilon
^{\mu }$ of the ray as a unit vector along the major principal axis
of the coherency tensor $I_{\mu \nu }$ of the radiation.\label{XRef-Definition-121084329}
\end{definition}
\begin{definition}

For a ray of composite electromagnetic radiation with a geometric
optics frequency spectrum and arbitrary state of polarization and
with transverse polarization vector $\epsilon ^{\mu }$, define the
{\bfseries {\itshape polarization wiggle rate}} as the scalar
\begin{equation}
\psi \equiv \epsilon ^{\gamma *}\epsilon _{\gamma \mu }\mathcal{D}_{u}\epsilon
^{\mu },%
\label{XRef-Equation-12108278}
\end{equation}

and the {\bfseries {\itshape mean helicity phase wiggle rate}} as
the scalar
\begin{equation}
\chi \equiv -i \epsilon ^{\mu *}\mathcal{D}_{u}\epsilon _{\mu }.
\end{equation}
\end{definition}
\begin{remark}

Defined this way, $\epsilon ^{\mu }$ is not uniquely defined, as
$e^{i \varphi }\epsilon ^{\mu }$ will also be a choice for any real
$\varphi $. However, it is evident from eq. (\ref{XRef-Equation-12108278})
that the definition of the polarization wiggle rate is\ \ independent
of the phase $\varphi $ and therefore uniquely defined.
\end{remark}
\begin{remark}

For a plane electromagnetic wave, the electric field vector is aligned
with the major principal axis of the coherency tensor, so the definition
of transverse polarization vector $\epsilon ^{\mu }$ given in Definition
\ref{XRef-Definition-121084329} is consistent with Definition \ref{XRef-Definition-121084119}.
\end{remark}
\begin{remark}

The helicity decomposition of the transverse polarization vector
$\epsilon ^{\mu }$ for a plane electromagnetic wave done in Section
\ref{XRef-Subsection-51152910} also applies to $\epsilon ^{\mu }$
defined according to Definition \ref{XRef-Definition-121084329}.
\label{XRef-Remark-11721320}
\end{remark}
\begin{remark}

In the following, we will compare quantities evaluated at different
time instants along the same timelike geodesic. We will do so informally,
for the sake of notational simplicity, but with the implicit understanding
that such comparisons are done by the use of push-forwards and pullbacks
along the geodesic. 
\end{remark}
\begin{proposition}

Assume the presence of a ray of partially coherent radiation, modeled
as an ensemble of plane waves, each emitted with constant polarization.
When evaluated at two events along the timelike geodesic of the
frame of observation, the transverse polarization vectors of the
components of this ensemble transform uniformly in the frame of
observation between these events by a transformation that can be
represented by a single unitary transormation $U_{\mu \nu }$. For
component $j$, the transverse polarization vectors ${\text{\boldmath
$\epsilon $}}_{(j)}^{\prime }$ and ${\text{\boldmath $\epsilon $}}_{(j)}$\ \ evaluated
at the events are therefore related by this unitary transformation
as\label{XRef-Proposition-121116255}
\[
{\text{\boldmath $\epsilon $}}_{\left( j\right) }^{\prime }=\text{\boldmath
$U$} {\text{\boldmath $\epsilon $}}_{\left( j\right) }.
\]
\end{proposition}
\begin{proof}

Since the radiation is emitted with constant polarization, the only
change to the polarization vector for plane wave component $j$ between
times $\tau $ and $\tau ^{\prime }$ is induced by gravity. Since
${\text{\boldmath $\epsilon $}}_{(j)}^{\prime }$ and ${\text{\boldmath
$\epsilon $}}_{(j)}$ are both unit vectors, they are related by
a unitary transformation.The equation that controls the evolution
of the transverse polarization vector, eq. (\ref{XRef-Equation-22594812}),
can be written in a frequency independent form by factoring out
the rest frame momentum, similar to how the polarization wiggle
equation can be written in the frequency independent form of Corollary
\ref{XRef-Corollary-12485247}. This shows that the evolution of
the transverse polarization vector ${\text{\boldmath $\epsilon $}}_{(j)}$
of each plane wave component is uniform, implying that the unitary
transfmation relating ${\text{\boldmath $\epsilon $}}_{(j)}^{\prime
}$ and ${\text{\boldmath $\epsilon $}}_{(j)}$ must be the same for
all plane wave components.
\end{proof}

This enables us to infer a similar proposition for the coherency
tensor:
\begin{proposition}

Assume the presence of a ray of partially coherent radiation emitted
with constant frequency spectrum and modeled as an ensemble of plane
waves, each emitted with constant polarization. When evaluated at
two events along the timelike geodesic of the frame of observation,
the corresponding values of the coherency tensor are related by
a unitary transformation $U_{\mu \nu }$ as\label{XRef-Proposition-1211165341}
\begin{equation}
I_{\mu \nu }^{\prime }=\frac{I^{\prime }}{I}{\left( \text{\boldmath
$U$} \text{\boldmath $I$} {\text{\boldmath $U$}}^{-1}\right) }_{\mu
\nu },
\end{equation}

where $I^{\prime }$ and $I$ are the intensities evaluated at the
two events, $\text{\boldmath $U$}$ is the unitary transformation
of\ \ Proposition \ref{XRef-Proposition-121116255} that relates
the polarization vectors of the plane wave components, and $\text{\boldmath
$I$}$ is the coherency tensor.
\end{proposition}
\begin{proof}

Using the expansion of $I_{\mu \nu }$ in terms of component intensities
given by eq. (\ref{XRef-Equation-5232107}), we get
\[
I_{\mu \nu }^{\prime }=\sum \limits_{j}{I_{\left( j\right) }^{\prime
}( \text{\boldmath $U$} {\text{\boldmath $\epsilon $}}^{\left( j\right)
}) }_{\mu }{\left( \text{\boldmath $U$} {\text{\boldmath $\epsilon
$}}^{\left( j\right) }\right) }_{\nu }^{*}=\sum \limits_{j}I_{\left(
j\right) }^{\prime } {\left( \text{\boldmath $U$} \text{\boldmath
$i$} {\text{\boldmath $U$}}^{-1}\right) }_{\mu \nu },
\]

where $\text{\boldmath $i$}$ denotes the relative coherency tensor.
The unitary transformation $\text{\boldmath $U$}$ can be factored
out of the sum on the right-hand side of the equation, yielding
\[
I_{\mu \nu }^{\prime }={\left( \text{\boldmath $U$}( \sum \limits_{j}I_{\left(
j\right) }^{\prime } \text{\boldmath $i$})  {\text{\boldmath $U$}}^{-1}\right)
}_{\mu \nu }.
\]

Since the radiation is emitted with constant frequency spectrum,
the intensity fraction of each component is constant, which implies
that $i_{(j)}^{\prime }=i_{(j)}$. This gives $\sum \limits_{j}I_{(j)}^{\prime
} \text{\boldmath $i$}=\frac{{I}^{\prime }}{I}\text{\boldmath $I$}$.
\end{proof}

The following proposition follows immediately.
\begin{proposition}

Assume the presence of a ray of partially coherent radiation emitted
with constant frequency spectrum. Let ${\text{\boldmath $I$}}^{\prime
}$ and {\bfseries I} denote the coherency tensor evaluated at two
different events along the geodesic of the frame of observation.
Then, if $\text{\boldmath $X$}$ is an eigenvector of $\text{\boldmath
$I$}$, $\text{\boldmath $U$} \text{\boldmath $X$}$ is an eigenvector
of ${\text{\boldmath $I$}}^{\prime }$, where $\text{\boldmath $U$}$
is the unitary transformation of Proposition \ref{XRef-Proposition-1211165341}
relating ${\text{\boldmath $I$}}^{\prime }$ to $\text{\boldmath
$I$}$. 
\end{proposition}
\begin{proof}

Since $\text{\boldmath $X$}$ is an eigenvector of $\text{\boldmath
$I$}$, it must satisfy\ \ $\text{\boldmath $I$} \text{\boldmath
$X$}=\lambda  \text{\boldmath $X$}$, where $\lambda $ is the eigenvaue.
Using Proposition \ref{XRef-Proposition-1211165341}, we get ${\text{\boldmath
$I$}}^{\prime }\text{\boldmath $U$} \text{\boldmath $X$}=\frac{I^{\prime
}}{I}\lambda  \text{\boldmath $U$} \text{\boldmath $X$}$, which
proves that $\text{\boldmath $U$} \text{\boldmath $X$}$ is an eigenvector
of ${\text{\boldmath $I$}}^{\prime }$ with eigenvalue $\lambda 
I^{\prime }/I$.
\end{proof}

Thus, as expected, the polarization axis of the ensemble, defined
as the major principal axis of the coherency tensor, transforms
in the same way, and with the same rate, as the polarization axes
of the components of the ensemble. This confirms the assertion made
initially in this section.

We can therefore conclude that the polarization wiggle equations,
eqs. (\ref{XRef-Equation-1130185728}) and (\ref{XRef-Equation-1029174955})
of Corollary \ref{XRef-Corollary-12485247}, and the corresponding
solutions found in Section \ref{XRef-Section-331204723} generalize
to a ray of electromagnetic radiation with arbitrary state of polarization
and frequency spectrum, provided it is emitted with constant frequency
spectrum and constant state of polarization. With the definitions
of the polarization wiggle rates stated in Lemma \ref{XRef-Lemma-1227028},
let us conclude this section by stating the transport equations
that follow from eqs. (\ref{XRef-Equation-1130185728}) and (\ref{XRef-Equation-1029174955}).
\begin{theorem}

Given a ray of electromagnetic radiation that is emitted with arbitrary,
but constant, frequency composition with a geometric optics spectrum
and arbitrary, but constant, state of polarization, the transport
equation for the polarization axis wiggle rate $\omega $ and the
mean helicity wiggle rate $\chi $ are\label{XRef-Theorem-1211224831}
\begin{gather}
\nabla _{n}\omega +\kappa  \omega =\mathcal{Z}%
\label{XRef-Equation-121122483}
\\\nabla _{n}\chi +\kappa  \chi =-\mathcal{V}\mathcal{Z},%
\label{XRef-Equation-117174555}
\end{gather}

while the transport equation for the circular polarization wiggle
rate $\alpha $ is
\begin{equation}
\nabla _{n}\alpha +\kappa  \alpha =0,%
\label{XRef-Equation-117174515}
\end{equation}

where $n^{\mu }\equiv u^{\mu }+{\hat{p}}^{\mu }$ is the scale invariant
null vector, while $\kappa $ is the projection of the expansion
tensor defined by eq. (\ref{XRef-Equation-12272315}).
\end{theorem}
\begin{remark}

When the radiation is emitted with constant polarization ($\alpha
( \lambda _{*}) =0$), the solution to eq. (\ref{XRef-Equation-117174515})
is $\alpha ( \lambda ) =0$. See Corollary \ref{XRef-Corollary-11721050}.
\end{remark}
\begin{remark}

We may now generalize Lemma \ref{XRef-Lemma-1227028} to a ray of
electromagnetic radiation that is emitted with arbitrary, but constant,
frequency composition with a geometric optics spectrum and arbitrary,
but constant, state of polarization. The reason is that all results
that this lemma were derived from have been generalized to composite
radiation with a geometric optics spectrum, see Remark \ref{XRef-Remark-11721320}.
Since we assume the radiation is emitted with constant polarization,
Corollary \ref{XRef-Corollary-11721050} applies as well.
\end{remark}

\subsection{Transport of the Circular Polarization Degree}

From Definition \ref{XRef-Definition-11674543}, the circular polarization
degree $\mathcal{V}$ can be extracted from the polarization tensor
by contraction with the screen rotator:
\begin{equation}
\mathcal{V}=i \epsilon ^{\mu \nu }P_{\mu \nu }.%
\label{XRef-Equation-51411042}
\end{equation}

The circular polarization degree is an observable scalar, and its
transport equation is given by the following theorem. 
\begin{theorem}

Given a ray of composite electromagnetic radiation with a geometric
optics frequency spectrum and arbitrary state of polarization, its
circular polarization degree $\mathcal{V}$ is constant along the
null geodesic:\label{XRef-Theorem-1211224734}
\begin{equation}
\nabla _{n}\mathcal{V}=0.%
\label{XRef-Equation-12865236}
\end{equation}
\end{theorem}
\begin{proof}

By differentiating eq. (\ref{XRef-Equation-51411042}) with $\mathcal{D}_{n}$,
eq. (\ref{XRef-Equation-12865236}) is obtained by the use of Proposition
\ref{XRef-Proposition-112219738} and Lemma \ref{XRef-Lemma-1271865}.
\end{proof}
\begin{remark}

It should be mentioned here, as a check on consistency of our approach,
that the circular polarization wiggle equation for a single plane
wave, eq. (\ref{XRef-Equation-124101546}), follows by differentiating
eq. (\ref{XRef-Equation-12865236}) with $\mathcal{D}_{u}$ and using
the commutation relation of eq. (\ref{XRef-Equation-22593549}).
\end{remark}

\section{Curvature Twist}\label{XRef-Section-112692639}

Polarization wiggling is sourced by curvature twist. Let us take
a closer look at this scalar and try to relate it to other known
quantities.

The definition of the curvature twist scalar $\mathcal{Z}$ is given
in Definition \ref{XRef-Definition-121263034} in terms of the Riemann
tensor. The curvature twist scalar contracts the Riemann tensor
with the four basis vectors of a tetrad adapted to a given null
geodesic. This tetrad consists of the two unit vectors $u^{\mu }$
and ${\hat{p}}^{\mu }$ as well as an arbitrary choice of polarization
basis, represented by two transverse unit vectors $e_{A}^{\alpha
}, A=1,2$. The two transverse directions are implicitly present
in the screen rotator $\epsilon ^{\rho \sigma }$, because it can
be expressed in terms of the two transverse basis vectors as $\epsilon
^{\alpha \beta }=e_{A}^{\alpha }e_{B}^{\beta }\varepsilon ^{\mathrm{AB}}$.
\begin{remark}

The screen rotator $\epsilon ^{\alpha \beta }=e_{A}^{\alpha }e_{B}^{\beta
}\varepsilon ^{\mathrm{AB}}$ is invariant with respect to change
of polarization basis $e_{A}^{\alpha }$. This follows from its definition
in Definition \ref{XRef-Definition-112264556}, which implies that
the curvature twist scalar $\mathcal{Z}$ is invariant with respect
to rotation of the polarization basis. 
\end{remark}

This section is devoted to a closer examination of the curvature
twist scalar $\mathcal{Z}$. We will show that it can be related
to other known quantities; the rotation of the reference geodesic
congruence $u^{\mu }$, the second Weyl scalar and the magnetic part
of the Weyl tensor.\ \ 

Let us start by expressing it in terms of the Weyl tensor.
\begin{lemma}

The curvature twist about a spatial unit vector ${\hat{p}}^{\mu
}$ can be expressed in terms of the Weyl tensor $C_{\gamma \rho
\beta \sigma }$ as\label{XRef-Lemma-121265327}
\begin{equation}
\mathcal{Z}\equiv u^{\gamma }{\hat{p}}^{\beta }\epsilon ^{\rho \sigma
}C_{\gamma \rho \beta \sigma }=\frac{1}{2}u^{\gamma }{\hat{p}}^{\beta
}\epsilon ^{\rho \sigma }C_{\gamma \beta \rho \sigma }.%
\label{XRef-Equation-1026182823}
\end{equation}
\end{lemma}
\begin{proof}

The proof of this is given in Appendix \ref{XRef-Subsection-11765546}.
\end{proof}

A conformally flat spacetime has vanishing Weyl tensor. The following
corollary then follows directly from Lemma \ref{XRef-Lemma-121265327}.
\begin{corollary}

Conformally flat spacetimes have vanishing curvature twist and do
not induce polarization wiggling.
\end{corollary}

Using the complex null tetrad, the curvature twist scalar can be
expressed in terms of the Weyl scalar $\Psi _{2}$:
\begin{corollary}

Given the complex null tetrad of Section \ref{XRef-Subsection-2269858},
the curvature twist scalar is related to the Weyl scalar $\Psi _{2}$
as follows:\label{XRef-Corollary-102874931}
\[
\mathcal{Z}=2 \operatorname{Im}[ \Psi _{2}] .
\]
\end{corollary}
\begin{proof}

Expanding eq. (\ref{XRef-Equation-1026182823}) in terms of the complex
null tetrad of Section \ref{XRef-Subsection-2269858} yields
\[
\mathcal{Z}=-i k^{\rho }l^{\mu }m_{+}^{\sigma }m_{-}^{\nu }C_{\rho
\mu \sigma \nu }.
\]

The corollary follows from using the permutation relation of the
Weyl tensor and that the Weyl scalar $\Psi _{2}=k^{\rho }m_{-}^{\nu
}m_{+}^{\sigma }l^{\mu }C_{\rho \nu \sigma \mu }$.
\end{proof}

\subsection{Curvature Twist and the Magnetic Part of the Weyl Tensor}

With the choice of the timelike reference geodesic congruence $u^{\mu
}$ (see Section \ref{XRef-Subsection-2269858}) follows a decomposition
of the Weyl tensor into electric and magnetic parts \cite{Stephani-2003,Hervik:2012jn}.
We will now show that the curvature twist $\mathcal{Z}$ is related
to the magnetic part of the Weyl tensor. The notation of Stephani
\cite{Stephani-2003} is used. 

The {\itshape right-dual Weyl tensor} is defined as
\[
C_{\mu \nu \gamma \delta }^{\sim }\equiv \frac{1}{2}{\epsilon }_{\gamma
\delta \alpha \beta }{C_{\mu \nu }}^{\alpha \beta }.
\]

The {\itshape gravitoelectric tensor} is the electric part of the
Weyl tensor and defined as
\[
E_{\nu \beta }\equiv C_{\mu \nu \alpha \beta }u^{\mu }{u}^{\alpha
},
\]

while the {\itshape gravitomagnetic tensor} is the magnetic part
of the Weyl tensor and defined in terms of the right-dual Weyl tensor
as
\begin{equation}
H_{\nu \beta }\equiv C_{\mu \nu \alpha \beta }^{\sim }u^{\mu }{u}^{\alpha
}.%
\label{XRef-Equation-1114222754}
\end{equation}
\begin{theorem}

For a given spacelike unit vector ${\hat{p}}^{\mu }$, the curvature
twist about ${\hat{p}}^{\mu }$ equals the projection of the gravitomagnetic
tensor along ${\hat{p}}^{\mu }$:\label{XRef-Theorem-121265825}
\begin{equation}
\mathcal{Z}=H_{\nu \beta }{\hat{p}}^{\nu }{\hat{p}}^{\beta }.%
\label{XRef-Equation-12126553}
\end{equation}
\end{theorem}
\begin{proof}

From Lemma \ref{XRef-Lemma-121265327} follows that the curvature
twist can be related to the right-dual Weyl tensor as follows:
\begin{equation}
\mathcal{Z}=u^{\mu }{u}^{\alpha }C_{\mu \nu \alpha \beta }^{\sim
}{\hat{p}}^{\nu }{\hat{p}}^{\beta }.%
\label{XRef-Equation-12126566}
\end{equation}

Eq. (\ref{XRef-Equation-12126553}) then follows directly from eq.
(\ref{XRef-Equation-12126566}) and the definition of the gravitomagnetic
tensor in eq. (\ref{XRef-Equation-1114222754}).
\end{proof}

Thus, the curvature twist depends exclusively on the gravitomagnetic
tensor.
\begin{corollary}

Spacetime regions with purely electric Weyl tensors $(H_{\nu \beta
}=0)$ do not induce polarization wiggling of passing electromagnetic
radiation.
\end{corollary}
\begin{proof}

This follows directly from Theorem \ref{XRef-Theorem-121265825}
and the definition of purely electric Weyl tensors.
\end{proof}

As we saw in Section \ref{XRef-Subsection-66111547}, previous results
have shown that the Gravitational Faraday Effect is present in spacetimes
with a non-zero gravitomagnetic field. Theorem \ref{XRef-Theorem-121265825}
generalizes this effect to any spacetime with a non-zero gravitomagnetic
tensor.
\begin{remark}

It is known that spacetimes with a purely electric Weyl tensor must
be of Petrov types I, D or O \cite{Stephani-2003}. Thus, spacetime
regions of Petrov types II, III and N will in general induce polarization
wiggling in passing electromagnetic radiation, because they will
have non-vanishing gravitomagnetic tensor. We must therefore expect
that gravitational radiation will induce polarization wiggling,
because spacetimes that exhibit gravitational radiation are of type
N. 
\end{remark}

\subsection{Curvature Twist and the Rotation of Rest Frame Geodesics}\label{XRef-Subsection-11146549}

Let $B_{\mu \nu }\equiv \nabla _{\nu }u_{\mu }$. Then\ \ $\omega
_{\mu \nu }\equiv B_{[\mu \nu ]}$ is the rotation of the geodesic
congruence $u^{\mu }$ \cite{Wald-General-Relativity,Poisson-2004}.
Next, we will show that curvature twist can be related to the rotation
$\omega _{\mu \nu }$ of the reference geodesic congruence $u^{\mu
}$. 

Let us first consider the scalar $\epsilon ^{\mu \nu }\omega _{\mu
\nu }$. Define $v\equiv -\epsilon ^{\mu \nu }\omega _{\mu \nu }$.
To help us understand what physical quantity $v$ is, let us evaluate
it in a linear gravity approximation with a Minkowski background.
In this approximation, we find that
\[
\epsilon ^{\mu \nu }\omega _{\mu \nu }=-\hat{p}\cdot \nabla \times
\text{\boldmath $u$},
\]

where $\text{\boldmath $u$}$ is the velocity of the inertial frame
relative to the coordinate frame of the background geometry. The
quantity $\nabla \times \text{\boldmath $u$} $ is then the vorticity
of the reference geodesic congruence $u^{\mu }$, analogous to the
vorticity of dust with velocity {\bfseries u.} Hence, $v=\hat{p}\cdot
\nabla \times \text{\boldmath $u$}$ is the vorticity of the reference
geodesic congruence about the direction of propagation of the polarized
radiation. 

Since the transverse projection $S_{\mu }^{\alpha }S_{\nu }^{\beta
}\omega _{\alpha \beta }$ must be proportional to $\epsilon _{\mu
\nu }$ and $\epsilon ^{\mu \nu }\epsilon _{\mu \nu }=2$, we find
that 
\[
S_{\mu }^{\alpha }S_{\nu }^{\beta }\omega _{\alpha \beta }=-\frac{1}{2}v
\epsilon _{\mu \nu }.
\]

In Section \ref{XRef-Subsection-43010640}, we saw that $\epsilon
_{\mu \nu }$ is the generator of rotations in the screen. It is
therefore natural to interpret $\frac{1}{2}v$ as the angular speed
of the geodesic congruence in the transverse plane, which is analogous
to the vorticity of a fluid being twice its angular velocity \cite{Tannehill-Fluid-Mechanics}.
This motivates the following definition:
\begin{definition}

For a given spacelike unit vector ${\hat{p}}^{\mu }$, define the
{\bfseries vorticity of the reference geodesic congruence} $u^{\mu
}$ about ${\hat{p}}^{\mu }$ as
\[
v\equiv -\epsilon ^{\mu \nu }\omega _{\mu \nu }.
\]
\end{definition}
\begin{theorem}

For a given spacelike unit vector ${\hat{p}}^{\mu }$, the curvature
twist $\mathcal{Z}$ about ${\hat{p}}^{\mu }$ is related to the vorticity
$v$ of the reference geodesic congruence $u^{\mu }$ about ${\hat{p}}^{\mu
}$ as\label{XRef-Theorem-121218188}
\begin{equation}
\mathcal{Z}=\nabla _{n}v,%
\label{XRef-Equation-1212175918}
\end{equation}

where $n^{\mu }\equiv u^{\mu }+{\hat{p}}^{\mu }$.
\end{theorem}
\begin{proof}

Let $B_{\mu \nu }\equiv \nabla _{\nu }u_{\mu }$. By Proposition
\ref{XRef-Proposition-112219738}, $\mathcal{D}_{\beta }\epsilon
_{\mu \nu }=0$. Therefore,
\begin{equation}
\epsilon ^{\mu \nu }\mathcal{D}_{n}B_{\mu \nu }=\nabla _{n}\left(
\epsilon ^{\mu \nu }\omega _{\mu \nu }\right) ,%
\label{XRef-Equation-121217532}
\end{equation}

where $\omega _{\mu \nu }\equiv B_{[\mu \nu ]}$ is the rotation
of the geodesic congruence $u^{\mu }$ \cite{Wald-General-Relativity,Poisson-2004}.
On the other hand,
\begin{equation}
\epsilon ^{\mu \nu }n^{\alpha }\nabla _{\alpha }\nabla _{\nu }u_{\mu
}=\epsilon ^{\mu \nu }n^{\alpha }\nabla _{\nu }\nabla _{\alpha }u_{\mu
}+\epsilon ^{\mu \nu }R_{\alpha \nu \mu \beta }u^{\beta }n^{\alpha
}.%
\label{XRef-Equation-1212175741}
\end{equation}

The first term on the right-hand side vanishes, because
\begin{equation}
\epsilon ^{\mu \nu }n^{\alpha }\nabla _{\nu }\nabla _{\alpha }u_{\mu
}=n^{\alpha }\mathcal{D}_{\nu }\mathcal{D}_{\alpha }( \epsilon ^{\mu
\nu }u_{\mu }) =0,
\end{equation}

since $\epsilon ^{\mu \nu }u_{\mu }=0$. We recognize that the second
term on the right-hand side of eq. (\ref{XRef-Equation-1212175741})
is $-\mathcal{Z}$. Equating eqs. (\ref{XRef-Equation-121217532})
and (\ref{XRef-Equation-1212175741}) yields eq. (\ref{XRef-Equation-1212175918}),
which proves the relationship between curvature twist and the rotation
of the geodesic congruence.
\end{proof}

A corollary follows directly.
\begin{corollary}

Polarization wiggling vanishes when the reference geodesic congruence
is hypersurface orthogonal.
\end{corollary}
\begin{proof}

This follows from Theorem \ref{XRef-Theorem-121218188} and the fact
that timelike hypersurface orthogonal geodesic congruences have
zero rotation $(\omega _{\mu \nu }=0$) \cite{Poisson-2004}.
\end{proof}

\subsection{Polarization Wiggling and Vorticity}

Theorem \ref{XRef-Theorem-121218188} expresses the curvature twist
scalar in terms of the vorticity of the timelike reference geodesic
congruence. This allows us to reexpress the polarization axis wiggle
equation for a plane electromagnetic wave, eq. (\ref{XRef-Equation-1249368})
of Lemma \ref{XRef-Lemma-12147341}.
\begin{corollary}

For a given plane electromagnetic wave with wave vector $p^{\mu
}$ and rest frame momentum $p$, the transport equation for the polarization
axis wiggle rate $\omega $ is
\begin{equation}
\nabla _{p}\left( \frac{\omega -v}{p}\right) =-\kappa  v .%
\label{XRef-Equation-121471157}
\end{equation}
\end{corollary}
\begin{proof}

Eq. (\ref{XRef-Equation-121471157}) follows when applying Theorem
\ref{XRef-Theorem-121218188} to the polarization axis wiggle equation
of Lemma \ref{XRef-Lemma-12147341} and using eq. (\ref{XRef-Equation-58123722})
as well as the definition of $\kappa $ in eq. (\ref{XRef-Equation-12272315}).
\end{proof}

Eq. (\ref{XRef-Equation-121471157}) relates the polarization axis
wiggle rate to the kinematic quantities $\kappa $ and $v$ of the
reference geodesic congruence. This equation can readily be integrated,
yielding the following corollary.
\begin{corollary}

For a given plane electromagnetic wave with wave vector $p^{\mu
}$ and rest frame momentum $p$, the general solution to the polarization
axis wiggle equation is
\begin{equation}
\omega ( \lambda ) =v( \lambda ) +\frac{p( \lambda ) }{p( \lambda
_{*}) }\left( \omega ( \lambda _{*}) -v( \lambda _{*}) \right) 
- p( \lambda ) \text{}\operatorname*{\int }\limits_{\lambda _{*}}^{\lambda
}d\lambda ^{\prime } \kappa ( \lambda ^{\prime })  v( \lambda ^{\prime
}) .%
\label{XRef-Equation-12147406}
\end{equation}
\end{corollary}

$\kappa ( \lambda ^{\prime }) $ is the projection of the expansion
tensor defined by eq. (\ref{XRef-Equation-12272315}). Eq. (\ref{XRef-Equation-12147406})
is equivalent to eq. (\ref{XRef-Equation-1249239}); they are just
expressed differently. While in eq. (\ref{XRef-Equation-1249239}),
the polarization axis wiggle rate is sourced by the curvature along
the radiation path, in eq. (\ref{XRef-Equation-12147406}) it is
sourced by the vorticity of the reference geodesic congruence along
its path.

If we consider linear gravity with a Minkowski background, possibly
the simplest case:, we may evaluate $\kappa $ and $p$ on the background.
In this case, $\kappa =0$ and $p( \lambda ) =\mathrm{const}$ to
lowest order. The polarization axis wiggle rate then evaluates to
\[
\omega ( \lambda ) =\omega ( \lambda _{*}) +v( \lambda ) -v( \lambda
_{*}) .
\]

Thus, in linear gravity with a Minkowski background, the relationship
between the polarization axis wiggle rate and vorticity of the reference
geodesic congruence is particularly simple: The change in the polarization
axis wiggle rate between emission and detection equals the change
in vorticity between these events. For radiation emitted with constant
polarization ($\omega ( \lambda _{*}) =0$), the polarization axis
wiggle rate simply measures the change in vorticity between emission
and detection.

\subsection{Rotation and the Gravitomagnetic Tensor}

Theorem \ref{XRef-Theorem-121265825} and Theorem \ref{XRef-Theorem-121218188}
yield two independent expressions of the curvature twist scalar.
Equating them will relate the vorticity of the reference geodesic
congruence to the gravitomagnetic tensor.
\begin{corollary}

For a given null geodesic with tangent vector $p^{\mu }=p( u^{\mu
}+{\hat{p}}^{\mu }) $, the vorticity $v$ of the timelike reference
geodesic congruence $u^{\mu }$ about ${\hat{p}}^{\mu }$ is related
to the gravitomagnetic tensor $H_{\mu \nu }$ as follows:\label{XRef-Corollary-1214173952}
\begin{equation}
\frac{1}{p}\nabla _{p}v=H_{\mu \nu }{\hat{p}}^{\mu }{\hat{p}}^{\nu
}.
\end{equation}
\end{corollary}

From Corollary \ref{XRef-Corollary-1214173952} follows another result.
\begin{corollary}

If the reference geodesic congruence $u^{\mu }$ has constant vorticity
$v$ along every null geodesic, the\ \ spacetime must be purely electric
with respect to $u^{\mu }$ ($H_{\mu \nu }=0$). Conversly, if the
spacetime is purely electric with respect to $u^{\mu }$, the vorticity
$v$ of the reference geodesic $u^{\mu }$ must be constant along
any null geodesic. \label{XRef-Corollary-121418113}
\end{corollary}
\begin{proof}

Since $H_{\mu \nu }$ is a spatial tensor $(H_{\mu \nu }u^{\nu }=0$)
\cite{Stephani-2003}, $H_{\mu \nu }{\hat{p}}^{\mu }{\hat{p}}^{\nu
}=0$ for any spatial direction ${\hat{p}}^{\mu }$ if and only if
$H_{\mu \nu }=0$. 
\end{proof}
\begin{corollary}

 Given a timelike geodesic congruence $u^{\mu }$ that is irrotational
($\omega _{\mu \nu }=0$), and therefore hypersurface orthogonal,
the spacetime must be purely electric with respect to $u^{\mu }$
($H_{\mu \nu }=0$).
\end{corollary}
\begin{proof}

This follows directly from Corollary \ref{XRef-Corollary-121418113},
since an irrotational congruence has zero vorticity $v$.
\end{proof}

\section{Result Summary}\label{XRef-Section-527172254}

Let us give a brief summary of the results of this paper. We have
considered {\itshape polarization wiggling}, an observable effect
of a gravitational field on the state of polarization of electromagnetic
radiation with a geometric optics frequency spectrum. Two observable
scalars are conserved along the null geodesic of the radiation:
\begin{itemize}
\item {\itshape Polarization degree} $\mathcal{P}$: Its transport
equation given by Theorem \ref{XRef-Theorem-127181215}. 
\item {\itshape Circular polarization degree} $\mathcal{V}$: Its
transport equation given by eq. (\ref{XRef-Equation-12865236}) of
Theorem \ref{XRef-Theorem-1211224734}.
\end{itemize}

A gravitational field affects polarized electromagnetic radiation
by inducing a change in the state of polarization of the radiation
as it traverses the gravitational field. The effect is independent
of the radiation frequency, and can be quantified by three scalars
and their corresponding transport equations:
\begin{itemize}
\item {\itshape Polarization axis wiggle rate} $\omega $: It is
a linear combination of the rates of change of the helicity phases
in an inertial frame of observation. For linear polarization, it
is the angular speed of the polarization axis. Its transport equation
is eq. (\ref{XRef-Equation-121122483}) of Theorem \ref{XRef-Theorem-1211224831}.
\item {\itshape Circular polarization wiggle rate} $\alpha $: It
is half the rate of change of the circular polarization degree of
the radiation, as measured by an inertial observer. $ \alpha $ is
an observable scalar. Its transport equation is eq. (\ref{XRef-Equation-117174515})
of Theorem \ref{XRef-Theorem-1211224831}.
\item {\itshape Mean helicity phase wiggle rate} $\chi $:\ \ It
is the mean rate of change of the helicity phases in an inertial
frame of observation.\ \ Its transport equation is eq. (\ref{XRef-Equation-117174555})
of Theorem \ref{XRef-Theorem-1211224831}.
\end{itemize}

The wiggle rates were defined from transverse projections of the
time derivative of the transverse polarization vector, $\nabla _{u}\epsilon
^{\mu }$, which implies that they quantify rates of change of the
state of polarization in an inertial frame. See Section \ref{XRef-Subsection-58122719}
for more details. Three observable scalars are linear combinations
of\ \ $\omega $ and $\chi $:
\begin{itemize}
\item {\itshape Positive helicity phase wiggle rate} ${\overset{\cdot
}{\phi }}_{+}$: It is the rate of change of the phase of the positive
helicity component in an inertial frame of observation. Its relationship
to the wiggle rates $\omega $ and $\chi $ is given by Lemma \ref{XRef-Lemma-1227028}.
\item {\itshape Negative helicity phase wiggle rate} ${\overset{\cdot
}{\phi }}_{-}$: It is the rate of change of the phase of the negative
helicity component in an inertial frame of observation. Its relationship
to the wiggle rates $\omega $ and $\chi $ is given by Lemma \ref{XRef-Lemma-1227028}.
\item {\itshape Polarization axis rotation rate} $\overset{\cdot
}{\Omega }$: It is the rotation rate of the polarization axis in
an inertial frame of observation. Its relationship to the two helicity
phase wiggle rates is given by eq. (\ref{XRef-Equation-11721125}),
and its relationship to the wiggle rates $\omega $ and $\chi $ is
given by Lemma \ref{XRef-Lemma-1227028}. 
\end{itemize}

These three observables and the circular polarization wiggle rate
$\alpha $ quantify rates of change of the state of polarization.
We refer to them as the {\itshape polarization wiggling observables}.

Along with the conservation of the polarization degree $\mathcal{P}$
and circular polarization degree $\mathcal{V}$, the transport equations
for the wiggle rates $\omega $ and $\chi $ with the corresponding
line-of-sight solutions provide a complete representation of the
effects of gravity on the state of polarization of a ray of electromagnetic
radiation with a geometric optics frequency spectrum. By the use
of Lemma \ref{XRef-Lemma-1227028} and Corollary \ref{XRef-Corollary-11721050},
it is straight forward to derive the polarization wiggling observables
from these solutions. This representation is valid for any 4-dimensional
spacetime and for any composition of electromagnetic radiation with
a geometric optics frequency spectrum.

It is worth noting that, while the state of polarization, represented
by the coherency tensor, is a tensorial quantity, and therefore
varies with the choice of polarization basis, any {\itshape change}
to this state induced by gravity can be quantified in terms of the
scalar observables referenced here, which are independent of the
choice of polarization basis in the frame of observation.

Polarization wiggling is sourced by the curvature twist scalar,
$\mathcal{Z}$, defined in Definition \ref{XRef-Definition-121263034}.
In the study of the curvature twist scalar in Section \ref{XRef-Section-112692639},
curvature twist was related to three other quantities:
\begin{itemize}
\item The {\itshape gravitomagnetic tensor} (the magnetic part of
the Weyl tensor), $H_{\mu \nu }$. Curvature twist is the projection
of the gravitomagnetic tensor along the direction of propagation.
This relationship is given in Theorem \ref{XRef-Theorem-121265825}.
\item The {\itshape second Weyl scalar} ($\Psi _{2}$). This relationship
is given in Corollary \ref{XRef-Corollary-102874931}.
\item The {\itshape rotation} $\omega _{\mu \nu }$ {\itshape of
the timelike reference geodesic congruence}. This relationship is
given in Theorem \ref{XRef-Theorem-121218188}.
\end{itemize}

A consequence of the gravitomagnetic tensor $H_{\mu \nu }$ and the
rotation $\omega _{\mu \nu }$ both being related to curvature twist
is that they by association can be related to each other. This relationship
is given in Corollary \ref{XRef-Corollary-1214173952}.

\section{Discussion}\label{XRef-Section-65111418}

In the introduction, it was claimed that, given choices of timelike
geodesics of the observer and emitter, the identified observables
are unambiguously defined. However, when considering how the polarization
axis wiggle rate can be calculated from eq. (\ref{XRef-Equation-1249239}),
this quantity appears to depend on the choice of timelike reference
geodesics between emission and observation. Since the choice of
timelike reference geodesics between the two fixed geodesics of
the observer and emitter is arbitrary, the polarization axis wiggle
rate seems to be ambiguous when only considering eq. (\ref{XRef-Equation-1249239}),
seemingly contradicting the claim for it to be unambiguously defined.
So, on what basis can we claim that the polarizaton wiggle rate
is independent of the choice of timelike reference geodesics? The
claim is based on the definition of the polarization wiggle rate
in eq. (\ref{XRef-Equation-112964427}). {\itshape By definition},
it only depends on the motion of the observer, given a gravitational
field and a field of electromagnetic radiation. This makes the dependency
on the choice of timelike reference geodesics in\ \ eq. (\ref{XRef-Equation-1249239})
akin to how the calculation of scalar physical quantities often
requires the choice of coordinates: Coordinates can be arbitrarily
chosen to calculate the quantity, but the end result is independent
of this choice. Thus, as long as the congruence is differentiable,
the timelike reference geodesics in between the timelike geodesics
of the emitter and observer can be chosen arbitrarily\ \ to calculate
the observables. 

The timelike reference geodesics, defined in Section \ref{XRef-Subsection-2269858},
constitute the timelike geodesic congruence used to define the spacetime
decomposition. An object traveling along one of these geodesics
is then referred to as being spacially at rest. Related to this
choice is an assumption that is worth highlighting: The observer
and the emitter of polarized radiation both travel along timelike
reference geodesics. What if this is not the case, and the observer
and/or emitter have significant motion relative to local inertial
rest frames? This question has so far not been discussed in the
present paper, but it deserves to be considered. Obviously, the
polarization wiggle rate is subject to time dilation at both ends
of the geodesic segment between the emission and the observation
events. The results of this paper can still be applied as long as
the state of polarization is evaluated in a local rest frame. This
implies that initial conditions set by the emitter must be converted
to the local rest frame, and the observable result must be converted
from the local rest frame to the frame of the observer. For instance,
in order to apply eq. (\ref{XRef-Equation-1249239}) to a moving
emitter, the scalar momentum and the polarization axis wiggle rate
of the emitted radiation must be converted to the local rest frame
by use of a local Lorentz transformation. Similarly, the resulting
observed polarization axis wiggle rate must be converted from the
local rest frame to the frame of a moving observer by use of a local
Lorentz transformation. Therefore, with these additional transformations
of initial conditions and results, the results of this paper can
be applied to emitters and observers with arbitrary motions relative
to the chosen rest frames.

Another implication of our results that is worth remarking is that
the effects of gravitational fields that are localized in space
or time will be recorded in the polarization of any electromagnetic
radiation traversing them and may therefore be transmitted with
the radiation over long ranges. To illustrate this effect, consider
a localized gravitational source in a spacetime with curvature twist
that is everywhere zero, except in a small neighborhood around the
gravitational source, and a ray of electromagnetic radiation traversing
this region. Assume Gaussian normal (synchronous) coordinates along
the timelike reference geodesic congruence $u^{\mu }$. Then coordinate
time $t$ can be used as an affine parameter along these geodesics,
and the null geodesic can be reparametrized in terms of $t$ using
$dt/\mathrm{d\lambda }=p$. For simplicity, let the curvature twist
along the null geodesic be represented as an instant angular acceleration
at $\lambda =\lambda _{0}$:
\[
\mathcal{Z}( \lambda ) ={\Delta \omega }_{0}\delta ( t( \lambda
) -t( \lambda _{0}) ) =\frac{{\Delta \omega }_{0}}{p( \lambda )
}\delta ( \lambda -\lambda _{0}) .
\]

$\Delta$$\omega _{0}$ is a constant angular speed increment, and
$\lambda _{0}$ labels the event when the photon passes the gravitational
source. The linear polariztion wiggle rate solution of eq. (\ref{XRef-Equation-22775842})
then takes the form
\[
\omega ( \lambda ) ={\left( 1+z( \lambda _{0}) \right) }^{-1}{\Delta
\omega }_{0},
\]

where $z( \lambda _{0}) \equiv p( \lambda _{0}) /p( \lambda ) -1$
is the redshift of the gravitational source. Thus, the effect of
a localized gravitational source with non-zero curvature twist,
represented by ${\Delta \omega }_{0}$, is preserved in the radiation
and redshifted proportionally to the redshift of the gravitational
source.

Gravity's effect on electromagnetic polarization is real, and we
have found how it can be quantified in terms of scalar observables
that are independent of the choice of polarization basis. This indicates
that there is a prospect of using electromagnetic polarization to
measure both local and remote gravitational fields, pending sufficiently
accurate measurement techniques.

\section{Conclusions}\label{XRef-Section-457621}

This paper has analyzed {\itshape polarization wiggling}, an observable
effect of gravity on polarized electromagnetic radiation in a 4-dimensional
spacetime with arbitrary geometry, for radiation with arbitrary
state of polarization and an arbitrary geometric optics frequency
spectrum (frequencies below the geometric optics limit have negligible
power). The effect was analyzed from the perspective of a single
inertial observer observing the polarization of radiation emitted
by a point-like inertial emitter. The key questions posed were:
Is this effect observable, and if it is, what are its observables?
We have analyzed the problem in full generality in the context of
classical electromagnetism in an arbitrary spacetime, using covariant
representations. The only assumption made regarding the composition
of the radiation is that geometric optics apply to its components.
Our results apply to any metric theory of gravity in 4 dimensions.

Two observable scalars are conserved along the null geodesic of
the radiation: The {\itshape polarization degree} $\mathcal{P}$
and the {\itshape circular polarization degree} $\mathcal{V}$.

Focusing on effects observable by a single inertial observer, we
showed how the presence of curvature along the null geodesic of
polarized electromagnetic radiation may induce changes in the observed
state of polarization. The effect, which we refer to as {\itshape
polarization wiggling}, is independent of radiation frequency, but
it varies with the degree of circular polarization. 

The effect can be quantified as follows: The {\itshape polarization
wiggle rates} are three scalars that emerge from independent transverse
projections of the time derivative of the transverse polarization
vector, $\nabla _{u}\epsilon ^{\mu }$, and therefore quantify rates
of change of the state of polarization in the frame of observation.
Their transport equations, the {\itshape polarization wiggle equations},
follow by applying basic Riemannian geometry. Expanding the polarization
wiggle rate scalars by the use of the null tetrad formalism reveals
their helicity decomposition, where the polarization wiggle rate
scalars are linear combinations of a set of observables; the {\itshape
polarization wiggling observables:}
\begin{itemize}
\item {\itshape Positive helicity phase wiggle rate} ${\overset{\cdot
}{\phi }}_{+}$: It is the rate of change of the phase of the positive
helicity component in an inertial frame of observation. Its relationship
to the polarization wiggle rates is given by Lemma \ref{XRef-Lemma-1227028}.
\item {\itshape Negative helicity phase wiggle rate} ${\overset{\cdot
}{\phi }}_{-}$: It is the rate of change of the phase of the negative
helicity component in an inertial frame of observation. Its relationship
to the polarization wiggle rates is given by Lemma \ref{XRef-Lemma-1227028}.
\item {\itshape Polarization axis rotation rate} $\overset{\cdot
}{\Omega }$: It is the rotation rate of the polarization axis in
an inertial frame of observation. Its relationship to the two helicity
phase wiggle rates is given by eq. (\ref{XRef-Equation-11721125}),
and its relationship to the polarization\ \ wiggle rates is given
by Lemma \ref{XRef-Lemma-1227028}. 
\item {\itshape Circular polarization wiggle rate} $\alpha $: It
is half the rate of change of the circular polarization degree of
the radiation. 
\end{itemize}

The polarization wiggling observables quantify time rates of change
of the state of polarization in the frame of observation. 

The polarization wiggling observables \{${\overset{\cdot }{\phi
}}_{\pm }, \overset{\cdot }{\Omega }, \alpha $\} and the polarization
state observables $\mathcal{P}$ and $\mathcal{V}$ form a complete
set of observables for gravitational effects on electromagnetic
polarization. These observables are independent of any choice of
local polarization basis in the frame of observation. The transport
equations of the polarizaton wiggle rates and the corresponding
line-of-sight solutions provide a complete representation of the
effect of gravity on the state of polarization of a ray of electromagnetic
radiation with a geometric optics frequency spectrum.

Polarization wiggling is sourced by curvature through a scalar that
we refer to as the {\itshape curvature twist}. The curvature twist
scalar is the contraction of the Riemann tensor or the Weyl tensor
with the four principal basis vectors of a tetrad adapted to the
null geodesic. It provides a unified representation of gravity's
effect on electromagnetic polarization for any 4-dimensional spacetime
and any metric theory of gravity. Curvature twist enters the transport
equations for the polarization wiggle rates, and their line-of-sight
solutions contain integrals of the curvature twist over the radiation
null geodesic from emission to detection.

We showed how curvature twist is related to three known quantities;
the magnetic part of the Weyl tensor, the second Weyl scalar ($\Psi
_{2}$) and the rotation of the timelike rest frame geodesic congruence.
Curvature twist is the projection of the magnetic part of the Weyl
tensor along the direction of propagation of the radiation. It can
also be expressed in terms of the vorticity of the rest frame geodesic
congruence about the direction of propagation of the radiation.

Polarization wiggling offers a view and represention of gravity's
effect on electromagnetic polarization that is different from and
independent of the Gravitational Faraday Effect. The main advantages
of this approach are its generality and the observables it provides
for the effects of gravity on electromagnetic polarization. Our
results apply to 4-dimensional spacetimes of any geometry and composite
radiation fields with geometric optics frequency spectra and in
any state of polarization. The results are valid for any metric
theory of gravity. 

A noticeable implication of these results is that the effect of
gravitational fields that are localized in time or space will be
recorded in the polarization of any electromagnetic radiation traversing
these fields and transmitted over long ranges. This opens the prospect
of using this effect to measure gravitational fields directly. Pending
sufficiently accurate polarization measurement techniques, it is
therefore important to recognize the reality of gravity's effect
on electromagnetic polarization and how to quantify it in terms
of observables.

\section{Declarations}

{\bfseries Funding} \ \ \ No funding was received for this work.

{\bfseries Author Contribution} \ All work for this paper was done
by the author (K. Tangen). 

{\bfseries Conflict of Interest} \ \ The author declares that there
are no conflict of interests.

{\bfseries Data Availability Statement} No data were used or produced
in this work.

\appendix

\section{Miscellaneous Derivations}

\subsection{Transverse Covariant Derivative of the Screen Rotator}\label{XRef-Subsection-22575619}

The covariant derivative of the Levi-Civita tensor vanishes \cite{Poplawski:2009fb,Carrol-2004}:
\[
\nabla _{\gamma }\epsilon _{\alpha \beta \mu \nu }=0
\]

Given a null vector $p^{\mu }$ with an expansion $p^{\mu }=p( u^{\mu
}+{\hat{p}}^{\mu }) $ in terms of a timelike unit vector $u^{\mu
}$ and a spacelike unit-vector ${\hat{p}}^{\mu }$, the transverse
covariant derivative of the screen rotator $\epsilon _{\mu \nu }$,
as defined by eq. (\ref{XRef-Equation-11237352}), expands as follows:
\[
\mathcal{D}_{\rho }\epsilon _{\mu \nu }=S_{\mu }^{\lambda }S_{\nu
}^{\gamma }\epsilon _{\alpha \beta \lambda \gamma }\nabla _{\rho
}u^{\alpha }{\hat{p}}^{\beta }+S_{\mu }^{\lambda }S_{\nu }^{\gamma
}\epsilon _{\alpha \beta \lambda \gamma }u^{\alpha }\nabla _{\rho
}{\hat{p}}^{\beta }
\]

Consider the first term, $S_{\mu }^{\lambda }S_{\nu }^{\gamma }\epsilon
_{\alpha \beta \lambda \gamma }\nabla _{\rho }u^{\alpha }{\hat{p}}^{\beta
}$:\ \ $S_{\mu }^{\lambda }S_{\nu }^{\gamma }\epsilon _{\alpha \beta
\lambda \gamma }\nabla _{\rho }u^{\alpha }{\hat{p}}^{\beta }\neq
0$ if and only if $u_{\alpha }\nabla _{\rho }u^{\alpha }\neq 0$.
However, 
\[
u_{\alpha }\nabla _{\rho }u^{\alpha }=\frac{1}{2}\nabla _{\rho }\left(
u_{\alpha }u^{\alpha }\right) =0
\]

since $u^{\alpha }$ is a unit vector. This implies that
\[
S_{\mu }^{\lambda }S_{\nu }^{\gamma }\epsilon _{\alpha \beta \lambda
\gamma }\nabla _{\rho }u^{\alpha }{\hat{p}}^{\beta }=0.
\]

By the same argument,
\[
S_{\mu }^{\lambda }S_{\nu }^{\gamma }\epsilon _{\alpha \beta \lambda
\gamma }u^{\alpha }\nabla _{\rho }{\hat{p}}^{\beta }=0.
\]

Therefore, the transverse covariant derivative of the screen rotator
vanishes:
\begin{equation}
\mathcal{D}_{\rho }\epsilon _{\mu \nu }=0.%
\label{XRef-Equation-22085013}
\end{equation}

\subsection{The Commutator $[p,u]$}\label{XRef-Subsection-22593618}

Let us evaluate the commutator $[p,u]$ for a null vector $p^{\mu
}$ and a timelike unit vector $u^{\mu }$. Its definition is
\begin{equation}
{\left[ p,u\right] }^{\beta }\equiv \nabla _{p}u^{\beta }-\nabla
_{u}p^{\beta }.%
\label{XRef-Equation-22585334}
\end{equation}

It can be decomposed into transverse and non-transverse terms by
using the screen projector:
\[
\nabla _{p}u^{\beta }-\nabla _{u}p^{\beta }=\mathcal{D}_{p}u^{\beta
}-\mathcal{D}_{u}p^{\beta }-u^{\beta }u_{\gamma }\nabla _{p}u^{\gamma
}+u^{\beta }u_{\gamma }\nabla _{u}p^{\gamma }+{\hat{p}}^{\beta }{\hat{p}}_{\gamma
}\nabla _{p}u^{\gamma }-{\hat{p}}^{\beta }{\hat{p}}_{\gamma }\nabla
_{u}p^{\gamma }.
\]

First, we find that the transverse terms vanish:
\begin{gather*}
\mathcal{D}_{p}u^{\beta }=S_{\alpha }^{\beta }\mathcal{D}_{p}u^{\alpha
}=\mathcal{D}_{p}( S_{\alpha }^{\beta }u^{\alpha }) -u^{\alpha }\mathcal{D}_{p}S_{\alpha
}^{\beta }=0
\\\mathcal{D}_{u}p^{\beta }=S_{\alpha }^{\beta }\mathcal{D}_{u}p^{\alpha
}=\mathcal{D}_{u}( S_{\alpha }^{\beta }p^{\alpha }) -p^{\alpha }\mathcal{D}_{u}S_{\alpha
}^{\beta }=0.
\end{gather*}

The remaining terms become
\begin{multline*}
-u^{\beta }u_{\gamma }\nabla _{p}u^{\gamma }+u^{\beta }u_{\gamma
}\nabla _{u}p^{\gamma }+{\hat{p}}^{\beta }{\hat{p}}_{\gamma }\nabla
_{p}u^{\gamma }-{\hat{p}}^{\beta }{\hat{p}}_{\gamma }\nabla _{u}p^{\gamma
}\\
=-\frac{p^{\beta }}{p}\nabla _{u}p+\left( \frac{p^{\beta }}{p}-u^{\beta
}\right) {\hat{p}}_{\gamma }\nabla _{p}u^{\gamma }.
\end{multline*}

Now,
\[
{\hat{p}}_{\gamma }\nabla _{p}u^{\gamma }=-\frac{\nabla _{p}p}{p},
\]

where $p\equiv -u_{\alpha }p^{\alpha }$ is the scalar momentum of
the null vector. Thus, we find that the commutator $[p,u]$ is
\begin{equation}
{\left[ p,u\right] }^{\beta }=-\frac{p^{\beta }}{p}\left( \nabla
_{u}p+\frac{\nabla _{p}p}{p}\right) +u^{\beta }\frac{\nabla _{p}p}{p}.%
\label{XRef-Equation-22593549}
\end{equation}

\subsection{Curvature Twist}\label{XRef-Subsection-9108468}

 Using that 
\[
\epsilon ^{\lambda \mu }R_{\mu \alpha \beta \gamma }\epsilon ^{\alpha
}=\epsilon ^{\lambda \mu }S_{\mu }^{\rho }S_{\alpha }^{\sigma }\epsilon
^{\alpha }R_{\rho \sigma \beta \gamma }=\frac{1}{2}\epsilon ^{\lambda
\mu }( S_{\mu }^{\rho }S_{\alpha }^{\sigma }-S_{\mu }^{\sigma }S_{\alpha
}^{\rho }) \epsilon ^{\alpha }R_{\rho \sigma \beta \gamma },
\]

we get by the use of Proposition \ref{XRef-Proposition-122083225}:
\[
\epsilon _{\lambda }^{*}\epsilon ^{\lambda \mu }R_{\mu \alpha \beta
\gamma }\epsilon ^{\alpha }=\frac{1}{2}\epsilon _{\lambda }^{*}\epsilon
^{\lambda \mu }\epsilon _{\mu \alpha }\epsilon ^{\rho \sigma }R_{\rho
\sigma \beta \gamma }\epsilon ^{\alpha }=\frac{1}{2}\epsilon _{\lambda
}^{*}( -S_{\alpha }^{\lambda }) \epsilon ^{\alpha }\epsilon ^{\rho
\sigma }R_{\rho \sigma \beta \gamma }.
\]

Then, by using that $\epsilon _{\lambda }^{*}S_{\alpha }^{\lambda
}\epsilon ^{\alpha }=1$ and applying the decomposition of $p^{\alpha
}$ from eq. (\ref{XRef-Equation-62782635}), using the antisymmetry
of the Riemann tensor in the two last indices, we obtain the following
identity:
\begin{equation}
\epsilon _{\gamma }^{*}\epsilon ^{\gamma \mu }R_{\mu \lambda \alpha
\beta }p^{\alpha }u^{\beta }\epsilon ^{\lambda }=-\frac{p}{2}\epsilon
^{\rho \sigma }R_{\rho \sigma \beta \gamma }{\hat{p}}^{\beta }u^{\gamma
}.%
\label{XRef-Equation-225163422}
\end{equation}

The Riemann tensor has several index symmetries. Among these are
the antisymmetry of the first and last two indices
\begin{equation}
R_{\mu \alpha \beta \gamma }=R_{\left[ \mu \alpha \right] \beta
\gamma }=R_{\mu \alpha [ \beta \gamma ] },%
\label{XRef-Equation-22516335}
\end{equation}

the switching relation
\begin{equation}
R_{\mu \alpha \beta \gamma }=R_{\beta \gamma \mu \alpha }%
\label{XRef-Equation-225163346}
\end{equation}

and the permutation relation, which can be written
\begin{equation}
R_{\gamma \beta \rho \sigma }=-R_{\gamma \rho \sigma \beta }-R_{\gamma
\sigma \beta \rho }=R_{\gamma \rho \beta \sigma }-R_{\gamma \sigma
\beta \rho }.%
\label{XRef-Equation-22516420}
\end{equation}

The right-hand side of eq. (\ref{XRef-Equation-225163422}) can be
rewritten by using the antisymmetry of eq. (\ref{XRef-Equation-22516335})
and the switching relation of eq. (\ref{XRef-Equation-225163346}).
Eq. (\ref{XRef-Equation-225163422}) then takes the form
\begin{equation}
\epsilon _{\gamma }^{*}\epsilon ^{\mu \gamma }R_{\mu \lambda \alpha
\beta }p^{\alpha }u^{\beta }\epsilon ^{\lambda }=\frac{p}{2}u^{\gamma
}{\hat{p}}^{\beta }\epsilon ^{\rho \sigma }R_{\gamma \beta \rho
\sigma }.%
\label{XRef-Equation-225163748}
\end{equation}

By using the permutation symmetry of eq. (\ref{XRef-Equation-22516420}),
we obtain the following identity:
\begin{equation}
\epsilon _{\gamma }^{*}\epsilon ^{\mu \gamma }R_{\mu \lambda \alpha
\beta }p^{\alpha }u^{\beta }\epsilon ^{\lambda }=p \mathcal{Z},
\end{equation}

where
\begin{equation}
\mathcal{Z}\equiv u^{\gamma }{\hat{p}}^{\beta }\epsilon ^{\rho \sigma
}R_{\gamma \rho \beta \sigma }=\frac{1}{2}u^{\gamma }{\hat{p}}^{\beta
}\epsilon ^{\rho \sigma }R_{\gamma \beta \rho \sigma }%
\label{XRef-Equation-11762721}
\end{equation}

is the curvature twist scalar.

\subsection{Curvature Twist in Terms of the Weyl Tensor}\label{XRef-Subsection-11765546}

 $\mathcal{Z}$ can be expressed in terms of the Weyl tensor by applying
the relationship between the Weyl and Riemann tensors \cite{Carrol-2004}:
\[
R_{\rho \sigma \mu \nu }=C_{\rho \sigma \mu \nu }+\left( g_{\rho
[\mu }R_{\left. \nu \right] \sigma }-g_{\sigma [\mu }R_{\left. \nu
\right] \rho }\right) -\frac{1}{3}g_{\rho [\mu }g_{\left. \nu \right]
\sigma }R,
\]

where $R$ is the Ricci scalar. Let us first evaluate the right-hand
side of eq. (\ref{XRef-Equation-11762721}) term by term. First,
\begin{multline*}
u^{\rho }{\hat{p}}^{\mu }\epsilon ^{\sigma \nu }( g_{\rho [\mu }R_{\left.
\nu \right] \sigma }-g_{\sigma [\mu }R_{\left. \nu \right] \rho
}) =\\
\frac{1}{2}\left( u^{\rho }{\hat{p}}^{\mu }g_{\rho \mu }\epsilon
^{\sigma \nu }R_{\nu \sigma }-{\hat{p}}^{\mu }\epsilon ^{\sigma
\nu }{u}_{\nu }R_{\mu \sigma }-u^{\rho }\epsilon ^{\sigma \nu }{\hat{p}}_{\sigma
}R_{\nu \rho }+u^{\rho }{\hat{p}}^{\mu }\epsilon ^{\sigma \nu }g_{\sigma
\nu }R_{\mu \rho }\right) .
\end{multline*}

All terms on the right-hand side vanish because of the antisymmetry
and transversality of the screen rotator, $\epsilon ^{\sigma \nu
}$. Next,
\[
u^{\rho }{\hat{p}}^{\mu }\epsilon ^{\sigma \nu }\frac{1}{3}g_{\rho
[\mu }g_{\left. \nu \right] \sigma }R=\frac{R}{6}\left( u^{\rho
}{\hat{p}}^{\mu }g_{\rho \mu }\epsilon ^{\sigma \nu }g_{\nu \sigma
}-u^{\rho }{\hat{p}}_{\sigma }\epsilon ^{\sigma \nu }g_{\rho \nu
}\right) .
\]

Also here, the terms on the right-hand side vanish due to the antisymmetry
and transversality of the screen rotator, $\epsilon ^{\sigma \nu
}$. Thus, curvature twist can be expressed in terms of the Weyl
tensor as follows:
\begin{equation}
\mathcal{Z}\equiv u^{\rho }{\hat{p}}^{\mu }\epsilon ^{\sigma \nu
}R_{\rho \sigma \mu \nu }=u^{\rho }{\hat{p}}^{\mu }\epsilon ^{\sigma
\nu }C_{\rho \sigma \mu \nu }.
\end{equation}

It is then straight forward to prove the alternative expression
\begin{equation}
\mathcal{Z}=\frac{1}{2}u^{\gamma }{\hat{p}}^{\beta }\epsilon ^{\rho
\sigma }R_{\gamma \beta \rho \sigma }=\frac{1}{2}u^{\gamma }{\hat{p}}^{\beta
}\epsilon ^{\rho \sigma }C_{\gamma \beta \rho \sigma }.%
\label{XRef-Equation-1026182310}
\end{equation}

\section{Representing the State of Polarization of Electromagnetic
Radiation}
\label{XRef-AppendixSection-430133648}

An arbitrary state of polarization can be defined in terms of the
four Stokes parameters. We will now review the polarization representations
applied in this paper. It uses the nomenclature and definitions
of Section \ref{XRef-Section-44183620} above.

\subsection{The Polarization Matrix}

The polarization state of a classical field of electromagnetic radiation
is completely characterized by the four Stokes parameters \cite{Born-Wolf-Optics,Jackson-Electrodynamics}.
They are measurable quantities that can be expressed in terms of
the electric field vector and its projections against two mutually
orthogonal unit vectors that both are transverse to the direction
of propagation of the radiation, $\hat{p}$. Let $e_{A}^{\mu }, A=1,2$
denote the two transverse unit vectors, the {\itshape polarization
basis}. They form a basis for any transverse field. Since the electric
field vector $E^{\mu }$ is transverse to the direction of propagation,
it can be defined by its projection onto this basis: $E_{A}\equiv
e_{A}^{\mu }E_{\mu }$. We can then define the {\itshape coherency
matrix} $I_{\mathrm{AB}}$ in terms of $E_{A}$ as the expectation
value
\begin{equation}
I_{\mathrm{AB}}\equiv \left\langle  E_{A}E_{B}^{*}\right\rangle
\label{XRef-Equation-515165650}
\end{equation}

evaluated as a time average over a large number of radiation cylces
\cite{Born-Wolf-Optics}. We notice that the use of a complex representation
of the electric field in eq. (\ref{XRef-Equation-515165650}) is
permissible as long as the average is taken over a large number
of radiaton cycles, see Born and Wolf for details \cite{Born-Wolf-Optics}.\ \ \ $I_{\mathrm{AB}}$
is a $2\times 2$ Hermitian matrix. Written in matrix form, it can
be expressed in terms of the four Stokes parameters $I, Q, U, V$
as
\[
I_{\mathrm{AB}}=\frac{1}{2}{\left( \begin{array}{cc}
 I+Q & U-i V \\
 U+i V & I-Q
\end{array}\right) }_{\mathrm{AB}}.
\]

We define the polarization matrix $P_{\mathrm{AB}}$ as the traceless
part of the relative coherency matrix $i_{\mathrm{AB}}\equiv I_{\mathrm{AB}}/I$.
Written in terms of the relative Stokes parameters $\mathcal{Q}\equiv
Q/I$, $\mathcal{U}\equiv U/I$ and $\mathcal{V}\equiv V/I$, it takes
the form
\begin{equation}
P_{\mathrm{AB}}\equiv i_{\mathrm{AB}}-\frac{1}{2}\delta _{\mathrm{AB}}=\frac{1}{2}{\left(
\begin{array}{cc}
 \mathcal{Q} & \mathcal{U}-i \mathcal{V} \\
 \mathcal{U}+i \mathcal{V} & -\mathcal{Q}
\end{array}\right) }_{\mathrm{AB}}.%
\label{XRef-Equation-51785922}
\end{equation}

The polarization degree $\mathcal{P}\equiv \sqrt{\mathcal{Q}^{2}+\mathcal{U}^{2}+\mathcal{V}^{2}}$
can be expressed in terms of the polarization matrix as 
\begin{equation}
\mathcal{P}^{2}=2P^{\mathrm{BA}}P_{\mathrm{AB}}=2P^{\mathrm{AB}}P_{\mathrm{AB}}^{*}.%
\label{XRef-Equation-62820155}
\end{equation}

\subsection{The Coherency and Polarization Tensors}

Define the coherency tensor for an electromagnetic field, similar
to the definition of the coherency matrix of eq. (\ref{XRef-Equation-515165650}):
\begin{equation}
I_{\mu \nu }\equiv \left\langle  E_{\mu }E_{\nu }^{*}\right\rangle
.%
\label{XRef-Equation-515173059}
\end{equation}

The expectation value $\langle E_{\mu }E_{\nu }^{*}\rangle $ of
eq. (\ref{XRef-Equation-515173059}) is assumed to be a time average
over a large number of radiation cycles. The intensity $I$ is the
scalar
\[
I\equiv g^{\mu \nu }I_{\mu \nu }=S^{\mu \nu }I_{\mu \nu }.
\]

We retain the coherency matrix $I_{\mathrm{AB}}$ of\ \ eq. (\ref{XRef-Equation-515165650})
by projecting $I_{\mu \nu }$ with the polarization basis $e_{A}^{\mu
}$:
\[
I_{\mathrm{AB}}=e_{A}^{\mu }e_{B}^{\nu }I_{\mu \nu }.
\]

In the geometric optics limit, the amplitude $\mathcal{E}^{\mu }\equiv
i a p \epsilon ^{\mu }$ of a plane electromagnetic wave can be assumed
constant over a cycle. This implies that the Stokes parameters,
expressed in terms of time-averaged squares of the electric field,
$\langle E_{A}E_{B}^{*}\rangle $ are additive when adding up contributions
from multiple wave components. Therefore, the Stokes parameters
of an electromagnetic field expressed as a composition of plane
waves can be computed by summing up the Stokes parameters of each
component \cite{Born-Wolf-Optics}. This makes the Stokes parameters
useful for characterizing any radiation field of arbitrary composition
and arbitrary degree of polarization. Hence, the coherency tensor
$I_{\mu \nu }$ of eq. (\ref{XRef-Equation-515173059}) is a convenient
covariant and gauge invariant field representation of any electromagnetic
radiation field.

Define the relative coherency tensor $i_{\mu \nu }$ as
\[
i_{\mu \nu }\equiv \frac{I_{\mu \nu }}{I}.
\]

For a single plane wave with polarization vector $\epsilon ^{\mu
}$,  the relative coherency tensor is 
\[
i_{\mu \nu }^{\mathrm{pw}}( \epsilon ) =\epsilon _{\mu }\epsilon
_{\nu }^{*}.
\]

Following the definition of the polarization matrix in eq. (\ref{XRef-Equation-51785922}),
the polarization tensor $P_{\mu \nu }$ is defined as
\begin{equation}
P_{\mu \nu }=i_{\mu \nu }-\frac{1}{2}S_{\mu \nu }.%
\label{XRef-Equation-430132253}
\end{equation}

The polarization degree $\mathcal{P}$ can be expressed in terms
of the polarization tensor as
\begin{equation}
\mathcal{P}^{2}=2P^{\nu \mu }P_{\mu \nu }=2P^{\mu \nu }P_{\mu \nu
}^{*}.%
\label{XRef-Equation-510181610}
\end{equation}

\subsection{Stokes Parameters}\label{XRef-Subsection-4309528}

From eq. (\ref{XRef-Equation-51785922}), the polarization matrix
$P_{\mathrm{AB}}$ can be expanded in terms of the Pauli matrices
$\sigma _{a},a=1,2,3$. The polarization matrix then expands as 
\[
P_{\mathrm{AB}}=\frac{1}{2}\mathcal{S}_{a}\sigma _{\mathrm{AB}}^{a},
\]

where the relative Stokes parameters are $\mathcal{S}_{1}=\mathcal{U},
\mathcal{S}_{2}=\mathcal{V},\mathcal{S}_{3}=\mathcal{Q}$. Conversely,
the relative Stokes parameters can be found by contracting the polarization
matrix with the Pauli matrices:
\[
\mathcal{S}_{a}=\sigma _{a}^{\mathrm{AB}}P_{\mathrm{AB}}.
\]

Then, by selecting a transverse polarization basis $e_{\mu }^{A},
A=1,2$, we can define a Pauli tensor basis
\[
\sigma _{\mu \nu }^{a}=\sigma _{\mathrm{AB}}^{a}e_{\mu }^{A}e_{\nu
}^{B}.
\]

Based on the relationship between Pauli matrices, this basis satisfies
\[
{\left( \sigma ^{a}\sigma ^{b}\right) }_{\mu \nu }=\delta ^{\mathrm{ab}}S_{\mu
\nu }+i \varepsilon ^{\mathrm{abc}}\sigma _{\mu \nu }^{c}.
\]

This allows us to expand the polarization tensor in the Pauli tensor
basis as
\[
P_{\mu \nu }=\frac{1}{2}\mathcal{S}_{a}\sigma _{\mu \nu }^{a},
\]

Conversely, the relative Stokes parameters can be found by contracting
the polarization tensor with the Pauli tensors:
\[
\mathcal{S}_{a}=\sigma _{a}^{\mu \nu }P_{\mu \nu }.
\]

Let $s^{q}, q=Q,U$, denote the two symmetric Pauli matrices, with
$s^{Q}\equiv \sigma ^{3}$ and $s^{U}\equiv \sigma ^{1}$. Defining
$\varepsilon ^{\mathrm{QU}}=-\varepsilon ^{\mathrm{UQ}}=\varepsilon
^{312}=1$, the relationship between the two symmetric Pauli matrices
can be written
\[
{\left( s^{p}s^{q}\right) }_{\mu \nu }=\delta ^{\mathrm{pq}}S_{\mu
\nu }+ \varepsilon ^{\mathrm{pq}}\epsilon _{\mu \nu },
\]

and the antisymmetric Pauli matrix $\sigma ^{2}$ can be expressed
as 
\[
\sigma _{\mu \nu }^{2}=-i \epsilon _{\mu \nu }. 
\]

The polarization tensor then takes a form that separates the symmetric
and antisymmetric parts:
\begin{equation}
P_{\mu \nu }=-\frac{i}{2} \mathcal{V} \epsilon _{\mu \nu }+\frac{1}{2}S_{q}s_{\mu
\nu }^{q}, q=Q,U.%
\label{XRef-Equation-430131121}
\end{equation}

Eq. (\ref{XRef-Equation-430131121}) can be inverted to obtain the
circular polarization degree in terms of the polarization tensor:
\begin{equation}
\mathcal{V}=i \epsilon ^{\mu \nu }P_{\mu \nu }.%
\label{XRef-Equation-1029171427}
\end{equation}

\end{document}